\documentclass[a4paper, 11pt]{article}

\usepackage{amssymb, amsmath, amsthm}
\usepackage{fullpage}
\usepackage{graphicx}
\usepackage{url}
\usepackage{float}
\theoremstyle{plain}

\newtheorem{theorem}{Theorem}[section]
\newtheorem{lemma}[theorem]{Lemma}
\newtheorem{corollary}[theorem]{Corollary}
\newtheorem{proposition}[theorem]{Proposition}
\newtheorem{example}[theorem]{Example}

\numberwithin{equation}{section}

\theoremstyle{definition}
\newtheorem{definition}[theorem]{Definition}

\theoremstyle{remark}
\newtheorem{remark}[theorem]{Remark}

\newcommand{\R}{{\mathbb R}}

\newcommand{\Prob}{\mathop{\mathbb{P}}\nolimits}
\newcommand{\E}{\mathop{\mathbb{E}}\nolimits}

\newcommand{\sM}{{\mathcal M}}
\newcommand{\fpr}{{\mathit f}}
\newcommand{\pseq}{{\mathcal P}}
\newcommand{\bbarkap}{{\mathcal K}_c}
\newcommand{\Hb}{\overline{H}}

\everymath{\displaystyle}

\title{Model independent hedging strategies for variance swaps}

\author{David Hobson\thanks{
D.Hobson@warwick.ac.uk} \ 
and Martin Klimmek\thanks{M.Klimmek@warwick.ac.uk} \\
Department of Statistics, University of Warwick}

\date{\today}
\begin{document}
\maketitle

\begin{abstract}

A variance swap is a derivative with a path-dependent payoff which allows 
investors to take positions on the future variability of an asset. In the 
idealised setting of a continuously monitored variance swap written on an 
asset with continuous paths it is well known that the variance swap payoff can 
be replicated exactly using a portfolio of puts and calls and a dynamic 
position in the asset. This fact forms the basis of the VIX contract.

But what if we are in the more realistic setting where the contract is based 
on discrete monitoring, and the underlying asset may have jumps? We show that 
it is possible to derive model-independent, no-arbitrage bounds on the price 
of the variance swap, and corresponding sub- and super-replicating strategies. 
Further, we characterise the optimal bounds. The form of the hedges depends 
crucially on the kernel used to define the variance swap.

\end{abstract}

\section{Introduction}
\label{sec:intro}

The purpose of this article is to construct hedging strategies which 
super-replicate the payoff of a variance swap for {\em any} price path 
of the underlying asset, including price paths with jumps. The idea is 
that at initiation time $0$, an agent purchases a portfolio of puts and 
calls which she holds until time $T$. In addition, she follows a simple, 
dynamic investment strategy in the underlying over $[0,T]$. Then, for 
every possible path of the underlying, the sum of the payoff from the 
vanilla portfolio plus the gains from trade from the dynamic strategy is 
(more than) sufficient to cover the obligation from the variance swap. 
Implicit in this set-up is the idea that the super-hedge does not rely 
on any modelling assumptions. Instead, the super-hedge is robust 
even in the presence of jumps.

The problem of finding the cheapest super-hedging strategy can be seen as 
the dual of a primal problem which is to bound the prices for variance 
swaps over the class of all models for the asset price process which are 
consistent with the traded prices of puts and calls. If the variance 
swap is sold for the price upper-bound and hedged with the corresponding 
super-replicating strategy then the seller will not lose money under any 
scenario.

The model-independent approach should be contrasted with the standard 
methodology which begins with a stochastic model for asset prices, and 
then infers the price of the variation swap by calculating the expected 
payoff. However, in markets where vanilla instruments are liquidly 
traded, the prices of puts and calls contain information about the 
market's expectations of the future behaviour of asset prices. The 
existence of this information removes the need to model the future, and this 
fact forms the basis of the model-independent approach.

In addition to super-hedges and upper bounds on the price of the variance 
swap we also give sub-hedges and lower bounds. Moreover, our analysis is 
not restricted to any particular definition of the variance swap, nor is 
it based on a mathematical idealisation of a continuous time limit of 
the swap contract, but rather on a discrete set of observations. We 
define variance swaps through their kernels; bivariate functions with 
regularity properties making them suitable to measure variance 
properties of the price path. Examples of kernels include squared simple 
returns, squared $\log$ returns and squared price differences. 
Furthermore, the sub- and super-replicating hedges work for discretely 
sampled variance swaps and continue to work in the continuous time 
limit. As long as the price path has a quadratic variation, these limits 
exist by F\"{o}llmer's path-wise It\^{o} formula \cite{Follmer:79}.  The 
standard approach to variance swap pricing is to assume a stochastic 
model and that the underlying paths are generated from a semi-martingale 
process with respect to this model. In this article, a model is 
specified only when it is necessary to show that the cheapest 
super-replicating hedge is tight.

Under some minimal restrictions on the form of the variance swap kernel 
we find a family of super-hedging strategies. This family is 
parameterised by a set of monotone functions. Then, given that the 
prices of call options for the expiry date of the variance swap are 
known (or equivalently the marginal law of the underlying price process 
at maturity is known) we show that there exists a cheapest 
super-replicating hedge from the given family. This hedge is associated 
with a monotone function, and we use this function to describe a 
stochastic model for the forward price of the asset in which the price 
process is continuous, except perhaps for a single jump, after which the 
process remains constant. In the continuous time limit, the super-hedge 
replicates the payoff of the variance swap if the asset price follows 
this one-jump model. This shows that the bounds we produce are best 
possible and justifies the restriction of our search to hedging 
strategies within the given family.

This article shares the model-independent ethos for the pricing of 
variance swaps implicit in Neuberger \cite{Neuberger:1994} and Dupire 
\cite{Dupire:1992} in the setting of continuous price processes. In 
those articles, it was shown that if we assume that the asset price 
process is a continuous forward price, then the continuously monitored 
variance swap based on either squared $\log$ returns or squared simple 
returns is perfectly replicated by the following strategy: synthesise 
$-2 \log$ contracts using put and call options and trade continuously in 
the asset to hold a number of shares equal to twice the reciprocal of 
the current asset price at all times. We will refer to this strategy as 
the classical continuous hedge. By results due to Breeden and 
Litzenberger \cite{BreedenLitzenberger:78}, it is possible to 
approximate any sufficiently regular payoff with vanilla options.  As a 
special case Demeterfi et. al. \cite{DemeterfiDerman:992} show how to 
approximate the $\log$ contract with a finite range of vanilla options. 
It follows that in the setting of a continuous forward price, the unique 
no-arbitrage price for the variance swap is equal to the price of the 
contract with payoff equal to $-2 \log$ contracts. This result holds 
independently of any modelling assumptions beyond path continuity. The 
hedging strategies in this paper are of the same character, consisting 
of a static position in calls and puts and dynamic trading in the 
underlying. However, the underlying setup is considerably more general, 
and the results more powerful since the hedges continue to 
super-replicate the variance swap for discontinuous price-paths and 
discrete monitoring over arbitrary time partitions. Nonetheless, this 
increase in generality comes at a cost in that instead of a replicating 
strategy we get sub- and super-replicating strategies and instead of a 
unique no-arbitrage price we get a no-arbitrage interval of prices.

As is well known, the model-independent analysis of derivative prices is 
related to the construction of extremal solutions for the Skorokhod 
embedding problem. This relationship was first developed in 
Hobson~\cite{Hobson:98}, see Hobson~\cite{Hobson:10} for a recent 
survey, and exploits the idea that the classification of martingales with 
a given terminal law is equivalent to the classification of stopping 
times for Brownian motion, such that the stopped process has that given 
law. As we shall see, the monotone function which is associated with the 
cheapest super-hedging strategy arises in the Perkins solution 
\cite{Perkins:86} of the Skorokhod embedding problem 
\cite{Skorokhod:65}. For another example of model independent pricing and the 
connection between derivatives and the Skorokhod embedding problem in the 
context of variance options, see Cox and Wang~\cite{CoxWang:11}. In the 
setting of continuous price paths Cox and Wang~\cite{CoxWang:11} give bounds 
on the prices of call options on realised variance by exploiting a connection 
with the Root solution of the Skorokhod embedding problem.

In a recent paper \cite{Kahale:11}, Kahal\'e shows how to derive a tight 
sub-replicating strategy and corresponding model-independent lower bound for 
the price of a variance swap based on the squared $\log$ return kernel. The 
paper by Kahal\'e was an inspiration for our study which grew from an attempt 
to relate his work to the previous literature on model-independent bounds and 
the Skorokhod embedding problem. By framing the problem in this way we extend 
the results of Kahal\'e~\cite{Kahale:11} to other kernels, and give upper 
bounds as well as lower bounds. Moreover, in the case of squared returns where 
the connection is particularly explicit, we explain the origin of the extremal 
models, and we give a natural interpretation for some of the quantities 
appearing in \cite{Kahale:11} in terms of the Perkins embedding of the 
Skorokhod embedding problem. The analysis of the squared returns kernel 
motivates our general approach to variance swap bounds and links this work to 
previous results of the authors (Hobson and Klimmek~\cite{HobsonKlimmek:11a})
on characterising solutions of the Skorokhod embedding problem with particular 
optimality properties.

Also, we give an interpretation of the continuous 
time limit of the bounding strategies using F\"{o}llmer's path-wise It\^{o} 
calculus \cite{Follmer:79}. F\"{o}llmer's non-probabilistic It\^{o} calculus 
has been used elsewhere in mathematical finance, most notably by Bick and 
Willinger \cite{BickWillinger:1994}, and helps emphasise the fact that the 
gains from trade have an interpretation as (the limit of) Riemann sums.

One of the features of our analysis is that we study the variance swap 
under a variety of definitions for the contract. Early definitions of 
the variance swap were based on squared simple daily returns. 
Accordingly, the first analysis of the discrepancy between the classical 
continuous hedge and realised variance in the presence of jumps, which 
is due to Demeterfi et. al. \cite{DemeterfiDerman:99}, focused on this 
kernel. Later, the finance industry switched to a standardised 
definition based on $\log$-returns. (These contracts are typically sold 
OTC, and therefore any specification of the contract, and any 
observation frequency is possible.) In their comprehensive survey of the 
literature on variance derivatives Carr and Lee~\cite{CarrLee:20092} 
give a plausible reason for this change based on the fact that banks 
tended to be buyers of variance swaps. Conventional wisdom states that 
downward jumps are more frequent than upward jumps and, in contrast to 
the situation for squared simple returns, for the squared $\log$-return the 
contribution of downward jumps to the value of the variance swap is 
positive. Hence a switch to the $\log$ return definition was profitable 
to the banks.

This conjecture about the history of the variance swap illustrates the 
idea that in the presence of discrete monitoring or jumps (but not in 
the case of continuous monitoring and continuous price processes) each 
kernel lends different characteristics to variance swap values. Partly 
for this reason a variety of kernels have been proposed in the 
literature. Bondarenko~\cite{Bondarenko:07} introduces a kernel which 
lies between the squared $\log$ return and squared simple return 
definitions. Bondarenko's proposal is motivated by the fact that 
variance swaps based on this kernel can be replicated perfectly in the 
presence of jumps and in discrete time. In a recent working paper, 
Neuberger \cite{Neuberger:10} provides an alternative analysis for this 
type of payoff, introducing the so-called aggregation property. 
Neuberger also shows that kernels with this property have a 
model-independent price. The kernel proposed by Carr and Corso 
\cite{CarrCorso:01} in the context of commodity markets, which is based 
on squared price differences, belongs to the same class. Recently Martin 
\cite{Martin:11} has proposed yet another definition which is similar to 
the squared-return kernel but involves both the forward and the asset 
price. Our analysis covers all these kernels (though the kernel in 
\cite{Martin:11} is only covered for the case of zero interest rates), 
and emphasises that the impact of jumps depends crucially on the nature 
of the kernel. We find that kernels split into two classes - below we 
name them increasing and decreasing kernels --- and the special 
properties of the Bondarenko kernel come from the fact that it lies in 
the intersection of these classes.

Apart from asset price jumps, a further issue in the pricing and hedging 
of variance swaps is that the idealised continuous time limit may be a 
poor approximation to the traded contract which is based on discrete 
monitoring. For example, in \cite{BroadJain:08} Broadie and Jain show 
that when the price path has negative jumps the value of the 
discretely monitored ($\log$-return) variance swap can differ 
significantly from the value of the continuously monitored variance 
swap. Similarly, Bondarenko \cite{Bondarenko:07} investigates the 
hedging error that develops if the strategy of the classical continuous 
approach is approximated discretely, and reports replication errors of 
around 30 percent. From a theoretical perspective, Jarrow et. 
al.~\cite{JarrowProtter:10} show that we may have that the price of the 
continuously monitored variance swap is finite, whilst simultaneously the 
discretely 
sampled analogue may has an infinite price, an observation which raises 
fundamental questions about the validity of using the continuous time 
integrated variance as an approximation for the discretely monitored 
quantity. These previous studies underscore the importance of a 
model-independent analysis, especially one based on a finite number of 
monitoring points. Again in the continuous set-up, Platen and 
Chen~\cite{Platen:2010} show that variance swap values are infinite 
under realistic modelling assumptions and argue that this implies a risk 
of liquidity crises in financial markets. This article helps to quantify 
that risk: if call prices are such that the model independent upper 
bound for the variance swap is finite, then for all models which are 
consistent with the market data the variance swap value is finite.

Recognising the importance of the jump contribution to variance swap 
values, Carr, Lee and Wu \cite{CarrLeeWu:2009} 
show how it is possible to price and 
hedge a variance swap based on $\log$ returns if the asset price follows 
a L\'evy model. The analysis is extended to a more general class of 
variation swaps in \cite{CarrLee:10}. Given a particular L\'evy model 
for the dynamics of the price path, Carr and Lee show that there exists 
a model-dependent adjustment to the multiplier $2$ appearing in the 
classical continuous hedge such that the value of the variance swap is 
given by the new multiplier times the price of a $\log$-contract. In 
general, this price is not enforceable through a hedging strategy. 
Moreover, since all models are wrong and since the adjustment of the 
multiplier depends on specifying a particular model, this approach may 
still significantly mis-price realised variance, even if the L\'evy 
model calibrates well to options prices.

The appeal of the classical continuous hedge of Neuberger and Dupire is 
that, apart from price-path continuity, the only necessary assumption is 
that a $\log$ contract can be synthesised from put and call options, and 
then the option payoff can be replicated perfectly along each path. In 
this article, we continue to assume that regular payoffs can be 
replicated with vanilla options, but relax the continuity assumption. 
The prices of variance swaps are highly sensitive to the presence of 
jumps, and so this is an important advance.

The remainder of the paper is structured as follows.
In the next section we introduce the variance swap, and show how the 
definition depends on the form of the kernel. In 
Section~\ref{sec:motivation} we study the problem in the setting of 
continuous monitoring for a process with jumps. The understanding we 
develop in this section will motivate much of the subsequent analysis.
Section~\ref{sec:bds} contains the main result, and shows how to 
construct a class of sub-hedging strategies. In 
Sections~\ref{sec:expensive} and \ref{sec:cts} we find the most 
expensive sub-hedge of this class for a given set of call prices, and 
thus we derive a model independent bound on 
the price of a variance swap, and then we show this bound is best 
possible, by showing that in the continuous time limit it can be 
attained. 
In Section~\ref{sec:nzir} we extend our results from contracts written on 
forwards to include the case of contracts written on undiscounted prices.
The penultimate section gives some numerical results and 
concluding remarks are given in 
Section~\ref{sec:remarks}.

\section{Variance Swap Kernels and Model-Independent Hedging}
\label{sec:defs}

\subsection{Variation swaps}
We begin by defining the payoff of a variance swap on a path-wise 
basis. The payoff will depend on a kernel, on the times at which the 
kernel is evaluated and on the asset price at these times.

\begin{definition} \label{d:varswapkernel}
A {\it variation swap kernel} is a continuously differentiable bi-variate 
function 
$H:(0,\infty) \times (0,\infty) \rightarrow [0,\infty)$ such that for 
all $x \in (0,\infty)$, $H(x,x)=0=H_{y}(x,x)$. We say that the swap kernel is 
{\it regular} if it is 
twice continuously differentiable.

{\it A variance swap kernel} is
a regular variation swap kernel $H$ such that 
$H_{yy}(x,x) = x^{-2}$.
\end{definition}

Our main focus in this article is on variance swap kernels but we will 
discuss variation swap kernels $H^S(x,y) = (y-x)^3$ and 
$H^Q(x,y)=(y-x)^2$ briefly, see Remark~\ref{r:skew} and Example 
\ref{ex:oneprice}. (Strictly speaking $H^S$ is not a variation swap kernel 
since it is not non-negative, but most of our analysis still apllies in this 
case.)
A regular variation swap kernel is a variance swap kernel if
$H(x,x(1+\delta))= \delta^2 + o(\delta^2)$ for $\delta$ small.
Examples of variance swap kernels include 
$H^R(x,y)=\left(\frac{y-x}{x}\right)^2$, $H^L(x,y)=(\log(y)-\log(x))^2$ 
and $H^B(x,y)=-2\left(\log(y/x)-\left(\frac{y-x}{x}\right)\right)$.

\begin{definition} \label{d:partition}
A {\it partition} $P$ on $[0,T]$ is a set of times $0 = t_0 < t_1 < ... < t_N 
= T$. A partition is {\it uniform} if 
$t_k=\frac{kT}{N}$, $k=0,1,...N$. A sequence of partitions 
$\pseq=(P^{(n)})_{n \geq 1} = (\{t_k^{(n)}; 0 \leq k \leq N^{(n)}\})_{n 
\geq 1}$ is {\it dense} if $\displaystyle \lim_{n \uparrow \infty} \sup_{k 
\in \{0,...,N^{(n)}-1\}} |t^{(n)}_{k+1}-t^{(n)}_k|=0$.
\end{definition}

\begin{definition} \label{d:price}
A {\it price realisation} $\fpr=(\fpr(t))_{0 \leq t \leq T}$ 
is a c\`adl\`ag function $\fpr:[0,T] \rightarrow (0,\infty)$.
\end{definition}

\begin{definition} \label{d:payoff}
The {\it payoff of a variation swap} with kernel $H$ for a 
partition $P$ and a
price realisation 
$\fpr$ is 
\begin{equation}
\label{eq:varswapdef}
V_H(\fpr,P)=\sum_{k=0}^{N-1} 
H(\fpr(t_k),\fpr(t_{k+1})).
\end{equation}
\end{definition}

\begin{remark}
\begin{enumerate}
\item[(i)]
The price realisations $f$ should be interpreted as realisations of the 
forward price of the asset with maturity $T$. Later we will extend the 
analysis to cover un-discounted price processes, rather than forward 
prices.
\item[(ii)]
Large parts of the subsequent analysis can be extended to allow for 
price processes which can take the value zero,
provided we also define $H(0,0)=0$, or equivalently
truncate the sum in (\ref{eq:varswapdef}) at the first time in the
partition that $f$ hits 0.
In this case we must have 
that zero is absorbing, so that if $f(s)=0$, then $f(t)=0$ for all $s 
\leq t \leq T$.
\item[(iii)] 
In practice the variance swap contract is an exchange of the 
quantity $V = V_H(\fpr,P)$ for a fixed amount $K$. However, since 
there is no optionality to the contract, and since the contract paying 
$K$ can trivially be priced and hedged, we concentrate 
solely on 
the floating leg.
\item[(iv)]
In many of the earliest academic papers, and in particular in Demeterfi
et. al~\cite{DemeterfiDerman:99,DemeterfiDerman:992}, but also in some very 
recent 
papers, e.g. Zhu and Lian~\cite{ZhuLian:11},
the variance swap is defined in terms of the kernel $H^R$. 
However, it 
has become market practice to trade variance swaps based on the kernel 
$H^L$.  Nonetheless these contracts are traded over-the-counter and in 
principle it is possible to agree any reasonable definition for the 
kernel. Variance swaps defined using the variance kernel $H^B$ were 
introduced by Bondarenko~\cite{Bondarenko:07}, see also 
Neuberger~\cite{Neuberger:10}. As we shall see, the contract based on 
this kernel
has various desirable features. For continuous paths then in the  
limit of a dense partition 
the contract does not depend on the chosen kernel, 
see Example \ref{ex:oneprice} and Lemma~\ref{c:oneprice}, 
but this is not the case in general.
\item[(v)]
The labels $\{S,Q,R,L,B\}$ on the variation swap kernels denote
\{Skew, Quadratic, Returns, Logarithmic returns, Bondarenko\}
respectively.
\end{enumerate}
\end{remark}

Let $\pseq = (P^{(n)})_{n \geq 1}$ be a dense sequence of partitions.
If $\lim_{n \uparrow \infty} V_H(\fpr,P^{(n)})$ exists then the limit is
denoted $V_H(\fpr,P_\infty)$ and is called the continuous time limit of
$V_H(\fpr,P^{(n)})$ on $\pseq$.

An important concept will be the quadratic variation of a path. 
For a dense sequence of partitions $\pseq$, the {\it quadratic
variation} $[f]$ of $\fpr$ on $\pseq$ is defined to be $[\fpr]_t=\lim_{n
\uparrow \infty} \sum_{t^{(n)}_k \leq t}
(\fpr({t^{(n)}_{k+1}})-\fpr({t^{(n)}_k}))^2$, provided the limit exists.
We split the function into its continuous and discontinuous parts,
$[\fpr]_t=[\fpr]^c_t+\sum_{u \leq t} (\Delta \fpr(u))^2$.
Later we will relate this definition to that 
introduced by F\"{o}llmer~\cite{Follmer:79}, which is used to develop a 
path-wise version of It\^{o} calculus.

\subsection{Model independent pricing}
Our goal is to discuss how to price the variance swap contract, or more 
generally any path-dependent claim, under an 
assumption 
that European call and put (vanilla) options with maturity $T$ are traded and 
can be used for hedging, but without any assumption that a proposed model is a 
true reflection of the real dynamics. In this sense the strategies and prices 
we derive are model independent and robust.

Let call prices be given by $C(K)$, expressed in units of cash at time 
$T$. We assume that a continuum of calls are traded, and to preclude 
arbitrage we assume that $C$ is a decreasing convex function such that 
$C(0) = f(0)$, $C(K) \geq (f(0)-K)^+$ and $\lim_{K \uparrow \infty} C(K) 
= 0$, see e.g. Davis and Hobson~\cite{DavisHobson:07}. We exclude the 
case where $C(f(0))=0$ for then $C(K) = (f(0)-K)^+$ and the situation is 
degenerate: the forward price must remain constant and upper and lower 
bounds on the price of the variance swap are zero. Although we assume 
that calls are traded today (time 0), we do not make any assumption on 
how call prices will behave over time, except that they will respect 
no-arbitrage conditions and that on expiry they will be worth the 
intrinsic value.

\begin{definition}
A {\it synthesisable payoff} is a function $\psi : 
(0,\infty) \mapsto \R$ which 
can be represented as the difference of two convex functions (so that 
$\psi''(x)$ exists as a measure).
\end{definition}

Let $\Psi = \{ \psi: \psi \in \Psi \}$ be the set of synthesisable payoffs 
$\psi : (0,\infty) \mapsto \R$. Then we have
\begin{equation}
\label{eqn:staticpayoff}
\psi(f) 
= \psi(f(0)) 
+ \psi'_+(f(0))(f- f(0)) 
+\int_{(0,f(0)]} (x-f)^+ \psi''(x) dx
+\int_{(f(0),\infty)} (f-x)^+ \psi''(x) dx .
\end{equation}
where $\psi'_+$ denotes the right-derivative.
Thus we can represent the payoff of any sufficiently regular European 
contingent claim as a constant plus the gains from trade from holding a 
fixed quantity of forwards, plus the payoff of a static portfolio of 
vanilla calls and puts.

Let $D[0,t]$ denote the space of c\'adl\'ag functions on $[0,t]$.

\begin{definition}
A {\it dynamic strategy} for a fixed partition $P$ is a
collection of functions $\Delta = (\delta_{t_0}, \ldots, 
\delta_{t_{N-1}})$, 
where $\delta_{t_j} : D[0,t_j] \rightarrow \R$. 
The {\it payoff of a dynamic strategy} along a price realisation $f$ 
is
 \begin{equation} 
\label{eq:gainsfromtradecomplex}
\sum_{k=0}^{N-1} \delta_{t_k}((f(t))_{0 \leq t \leq t_k}) 
(f(t_{k+1})-f(t_k)). 
\end{equation}
Let $\bar{\Delta}(P)$ be the set of dynamic strategies.
\end{definition}

\begin{definition}
$\Delta = \bar{\Delta}(P)$ is a {\it Markov dynamic strategy} 
if $\delta_{t_j}(f(t)_{0 \leq t \leq t_j})=\delta_{t_j}(f(t_j))$ for all $j$. 
A Markov dynamic strategy is a 
{\it time homogeneous Markov dynamic strategy} (THMD-strategy) 
if $\delta_{t_j}(f(t_j))=\delta(f(t_j))$ for all $j$.
\end{definition}

In the sequel we will concentrate mainly on THMD-strategies. The 
quantity $\delta_{t_j}$ represents the quantity of forwards to be held 
over the interval $(t_j, t_{j+1}]$. In principle this quantity may 
depend on the current time and on the price history $(f(t))_{0 \leq t 
\leq t_j}$. However, as we shall see, for our purposes it is sufficient 
to work with a much simpler set of strategies where the quantity does 
not explicitly depend on time, nor on the price history except through 
the current value. We call this the Markov property, but note there are 
no probabilities involved here yet.

\begin{definition}  
A {\it semi-static hedging strategy} $(\psi,\Delta)$ is a function 
$\psi \in \Psi$ and a dynamic strategy $\Delta \in \bar{\Delta}(P)$. 
The terminal payoff of a semi-static hedging strategy for a price 
realisation $f$ is
\begin{equation}  
\label{eqn:payoff}
\psi(f(T)) + \sum_{k=0}^{N-1} \delta_{t_k}((f(t))_{0 \leq t \leq t_k}) 
(f(t_{k+1})-f(t_k)). 
\end{equation}
\end{definition}

Without loss of generality we may assume that $\psi'(f(0))=0$.
If not then we simply adjust each 
$\delta_{t_k}$ by the quantity $\psi'(f(0))$ and the payoff in 
(\ref{eqn:payoff}) is unchanged.  In the sequel, we will concentrate on 
the case when $\Delta$ is a THMD strategy. Then we identify $\Delta \in 
\bar{\Delta}(P)$ with $\delta:(0,\infty) \rightarrow \R$ and write 
$(\psi,\delta)$ instead of $(\psi,\Delta)$.

Given that investments in the forward market may be assumed to be costless, 
the dynamic strategy has zero price. Thus, in order to define the price of a 
semi-static hedging strategy it is sufficient to focus on the price associated 
with the payoff function $\psi$. The last two terms in 
(\ref{eqn:staticpayoff}) are expressed in terms of the payoffs of calls and 
puts. Thus we can identify the price of $\psi(f(T))$ with the price of a 
corresponding portfolio of vanilla objects. We also use put-call 
parity\footnote{This means that we do not need to introduce a notation for the 
put price, which is convenient since $P$ is already in use for the partition. 
Put-call parity for the forward says that the price of a put with strike $x$ 
is the price of a call with the same strike plus $f(0)-x$} to express the cost 
of the penultimate term in (\ref{eqn:staticpayoff}) in terms of call prices. 
Let $\Psi_0 = \{ \psi \in \Psi : \psi'_+(f(0))=0 \}$.

\begin{definition} \label{d:2.9semistatic}
The {\it price of a semi-static hedging strategy} $(\psi \in \Psi_0,\Delta 
\in \bar{\Delta}(P) )$ is \[ 
\psi(f(0)) + \int_{(0,f(0)]} \psi''(x) 
(C(x) + f(0) - x) dx + \int_{(f(0),\infty)} \psi''(x)C(x) dx .
\]
\end{definition}

The idea we wish to capture is that the agent holds a static position in calls
together with a
dynamic position in the underlying such that in combination they provide
sub- and super-hedges for the claim.

\begin{definition} \label{d:2.10subhedge}
Let $G = G((f(t_k))_{k = 0, \ldots N})$ 
be the payoff of a path-dependent option. Suppose that there 
exists a semi-static hedging strategy $(\psi, \Delta)$ 
such that on the partition $P$
\[ G \leq 
\mbox{(respectively $\geq$) } \; 
\psi(f(T)) + \sum_{k=0}^{N-1} 
\delta_{t_k}((f(t))_{0 \leq t \leq t_k}) (f(t_{k+1})-f(t_k))
.\]
Then $(\psi,\Delta)$ is called a {\it semi-static super-hedge} 
(respectively 
{\it semi-static sub-hedge}) for $G$.
\end{definition}

Given a semi-static sub-hedge (respectively super-hedge) we say that the price
of the sub-hedge (respectively super-hedge) is a {\it model independent lower 
(respectively upper) bound} on the price of the 
path-dependent claim $G$. 

\subsection{Consistent models}
The aim of the agent is to construct a hedge which works path-wise, and does 
not depend on an underlying model.
Nonetheless, sometimes it is convenient to introduce a probabilistic 
model and a stochastic process, and to interpret $f(t)$ as a realisation 
of that stochastic process.  In that case we work with a probability space 
$(\Omega, \mathcal{F}, \mathbb{F}, \Prob)$ supporting the stochastic 
process $X=(X_t)_{0 \leq t \leq T}$.

\begin{definition}
\label{d2.15}
A model $(\Omega, \mathcal{F}, \mathbb{F}, \Prob)$ and associated 
stochastic process $X=(X_t)_{0 \leq t \leq T}$ is {\it consistent} with the 
call prices $(C(K))_{K \geq 0}$ if $(X_t)_{t \geq 0}$ is a non-negative 
$(\mathbb{F}, 
\Prob)$-martingale and if
$\E[(X_T-K)^+] = C(K)$
for all $K>0$.
\end{definition}

In the setting of a stochastic model $V_{H}(X,P):\Omega \rightarrow 
\R^+$ is a random variable, and for $\omega \in \Omega$, 
$V_{H}(X(\omega),P)$ is a realised value of a variance swap. From a 
pricing perspective we are interested in getting upper and lower bounds 
on $\E[V_{H}(X(\omega),P)]$ as we range over consistent models. 
Knowledge of call prices is equivalent to knowledge of the marginal law 
of $X_T$ under a consistent model (Breeden and 
Litzenberger~\cite{BreedenLitzenberger:78}). If we write $\mu$ for the 
law of $X_T$ and if $C_{\mu}(K) = \E[(Z_{\mu}-K)^+]$ where $Z_\mu$ is a 
random variable with law $\mu$, then $X$ is consistent for the call prices 
$C$ if $C_\mu(K) = C(K)$.  
We write $m= 
\int_{0}^\infty x \mu(dx)$ 
and we assume, using the martingale property, that $f(0)=m$. 
Then the problem of characterising consistent models is equivalent to 
the problem of characterising all martingales with a given distribution 
at time $T$.

\section{Motivation}
\label{sec:motivation}

\subsection{The continuous case}
In the situation where both the monitoring and the price-realisations are 
continuous the theory 
for the pricing of variance swaps is complete and elegant. We will use 
this setting to develop intuition for the jump case.

Suppose that the price realisation $\fpr$ is continuous, 
and possesses a quadratic variation $[\fpr]:[0,T] \rightarrow \R^+$ on a 
dense sequence of partitions $\pseq$. Dupire \cite{Dupire:1992} and  
Neuberger \cite{Neuberger:1994} independently made the observation that 
the continuity assumption implies that a variance swap with payoff 
$\int_0^T \fpr(t)^{-2}d[\fpr]_t$ can be replicated perfectly 
by 
holding a static portfolio of $\log$ contracts and trading 
dynamically in the underlying asset. Both Dupire and Neuberger assume $\fpr 
\equiv X$ is a realisation of a semi-martingale, but
in our setting, the observation 
follows from a path-wise application of It\^{o}'s formula in the sense of 
F\"{o}llmer \cite{Follmer:79}, see Section \ref{sec:cts}. Applying 
It\^{o}'s formula to 
$-2\log(\fpr(t))$ we have
\begin{equation}
-2 \log(\fpr(T))+ 2 \log(\fpr(0)) = -2 \int_0^T \frac{1}{\fpr(t)} 
d\fpr(t) + \int_0^T \frac{1}{\fpr(t)^2}d[\fpr]_t.
\end{equation}
Then, as we show in Section \ref{sec:cts} below, down a 
dense 
sequence of partitions
\begin{equation} \label{eq:sec1:continuous}
V_{H}(\fpr,P_\infty)=
\int_0^T \frac{1}{\fpr(t)^2}d[\fpr]_t =
-2 \log(\fpr(T)) + 2 \log 
(\fpr(0)) + \int_0^T 
\frac{2}{\fpr(t)} 
d\fpr(t).
\end{equation}
Provided it is possible to trade continuously and without 
transaction costs, 
the right-hand-side of this identity has a clear interpretation as the 
sum of a European contingent claim with maturity $T$ and payoff $- 2 
\log ( \fpr(T)/\fpr(0))$ and the gains from trade from a dynamic 
investment
of $2/\fpr(t)$ in the underlying. 
Alternatively, the right-hand-side of (\ref{eq:sec1:continuous}) can be 
viewed as the payoff of a semi-static hedging strategy in the continuous 
time limit for the choice $\psi(x) = - 2 \log (x/f(0)) + 
2(x-f(0))/f(0)$ and $\Delta = (\delta_t)_{0 \leq t \leq T}$ where
$\delta_t( (f(u))_{0 \leq u \leq t})= (2/f(t)) - (2 /f(0))$. Note that 
there is equality in  
(\ref{eq:sec1:continuous}) so that $(\psi,\delta)$ is both a sub- and 
super-hedge for $V_H(f, P_\infty)$. 
In particular, under a price continuity assumption, the variance swap 
has a model-independent price and an associated riskless hedge.

\subsection{The effect of jumps on hedging with the classical continuous 
hedge}

Even if the continuity assumption cannot be justified, 
the associated replication strategy is nevertheless a reasonable 
candidate for a 
hedging strategy in the general case. Let us focus on the discrepancy between 
the payoff of the variance swap and the gains from trade resulting from 
using the hedge derived in the continuous case. 
The path-by-path It\^{o} 
formula continues to apply in the case with jumps, see 
\cite{Follmer:79} and Section \ref{sec:cts} below. Hence
\begin{eqnarray*} -2\log(\fpr(T))+ 2\log(\fpr(0)) & = & -2 \int_0^T 
\frac{1}{\fpr(t-)} 
d\fpr(t)+ \int_0^T \frac{1}{\fpr(t-)^2} d[\fpr]^c_t  \\
&& \hspace{5mm} + \sum_{0 \leq t \leq T} 
2 \left\{ \left(\frac{\Delta \fpr(t)}{\fpr(t-)}\right) 
- \log\left(1+\frac{\Delta \fpr(t)}{\fpr(t-)}\right) \right\}.
\end{eqnarray*}
Note that $d[\log(\fpr)]_t={d[\fpr]^c_t}/{\fpr(t-)^2}+
 (\Delta \log(\fpr(t)))^2$. By adding and subtracting the 
discontinuous part of the quadratic variation of $\log(\fpr)$ on the 
right-hand-side of the above expression, we find
\begin{equation} \label{eq:rearrange}
-2\log(\fpr(T)) + 2 \log f(0) = - 2 \int_0^T \frac{1}{\fpr(t-)} 
d\fpr(t)+ 
[\log(\fpr)]_T  
- \sum_{0 \leq t \leq T} J_L( {\Delta 
\fpr(t)}/{\fpr(t-)}) 
\end{equation}
where \[J_L(\eta)=-2\eta + 2\log(1+\eta)+\log(1+\eta)^2.\] 
It is intuitively clear, but see also Corollary~\ref{cor:ctslimit}, that 
$V_{H^L}(\fpr,P_\infty) \equiv [\log(\fpr)]_T$. Then it 
follows by re-arrangement of equation (\ref{eq:rearrange}) that the 
discrepancy between the realised value of the variance swap 
$V_{H^L}(\fpr,P_\infty)$ and the return 
generated by the classical continuous hedging strategy can be 
represented as 
the sum of the jump contributions:
\begin{equation*}
V_{H^L}(\fpr,P_\infty) - \left( - 2 \log(\fpr(T)) +2 \log \fpr(0) + 2 
\int_0^T 
\frac{1}{\fpr(t-)} 
d\fpr(t) \right) = \sum_{0 \leq t \leq T} 
J_L\left(\frac{\Delta \fpr(t)}{\fpr(t-)}\right).
\end{equation*}
We call this the hedging error with the convention that if 
the hedge sub-replicates the variance swap then the 
hedging error is positive.

Now consider the kernel $H^R$ and define $V_{H^R}(\fpr,P_\infty) = \int_0^T d 
[f]_t/f(t-)^2$, again, see Corollary~\ref{cor:ctslimit} for justification.
By a similar analysis, but adding and subtracting $\left(\frac{\Delta 
\fpr(t)}{\fpr(t-)}\right)^2$ instead of the discontinuous part of the 
quadratic variation of $\log(\fpr)$, 
we have
\begin{equation*}
V_{H^R}(\fpr,P_\infty) - \left( -2 \log(\fpr(T)) + 2 \log(\fpr(0)) + 2 
\int_0^T 
\frac{1}{\fpr(t-)}
d\fpr(t) \right) = \sum_{0 \leq t \leq T} J_R\left(\frac{\Delta
\fpr(t)}{\fpr(t-)}\right). \end{equation*}
where 
\[J_R(\eta)=- 2 \eta + 2 \log(1+\eta)+ \eta^2.\] 
In the continuous case, under some 
mild regularity conditions on $\fpr$ and $\pseq$, the variance swap 
value is independent of the chosen kernel. In contrast, the value of a 
variance swap in the general case is highly dependent on the chosen 
kernel. 

To see that this is the case, and to examine the impact of jumps on the 
hedging error for the kernels $H^L$ and $H^R$ we consider the shapes 
of the functions 
$J_R$ and $J_L$, see Figure 1.
For the kernel $H^L$, a downward jump results in a positive contribution
to the hedging error. Thus, if all jumps are downwards, then the 
classical 
continuous hedging
strategy sub-replicates $V_{H^L}(\fpr,P_\infty)$. Conversely, upward
jumps result in a negative  contribution
to the hedging error.
The story is reversed for the kernel
$H^R$.

\begin{figure}[H]\label{Fig.1}
\begin{center}
\includegraphics[height=5cm,width=7cm]{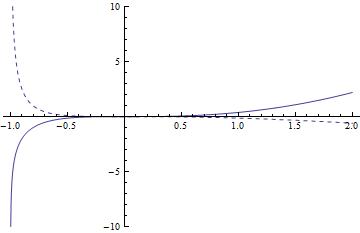}
\caption{$J_L$ (as represented by the dashed line) is convex decreasing for $x 
\leq 0$ 
and concave decreasing for $x \geq 0$. In contrast $J_R$ (solid line) 
is first concave increasing and then convex increasing. The different shapes 
of these two curves explains the different nature of the dependence of the 
payoff of the variance swap on upward and downward jumps for different 
kernels.}
\end{center}
\end{figure}

It follows from the argument in the previous paragraph that for the 
kernel $H^L$ the hedging error will be maximised under scenarios for which the 
price 
realisation has downward jumps, but no upward jumps. Paths with this 
feature might arise as realisations of $-N$ where $N=(N_t)_{t \geq 0}$ 
is a compensated Poisson process.
Moreover, from the convexity of $J_L$ on $(-1,0)$, it is plausible that 
the scenarios in which the hedging error is maximised are those in which price realisations have a single large downward jump, rather than a series of small jumps. Again 
if we wish to minimise the hedging error we should expect a single 
large upward jump, and the story is reversed for the kernel $H^R$.  
 
In summary, we find that, under a continuity assumption on $\fpr$, and 
for a dense sequence of partitions, the value of a variance swap is 
independent of the kernel and can be replicated with a static hedge in a 
forward contract and a dynamic hedging strategy. In the presence of 
jumps, however, the value of the variance swap depends on the kernel. An 
agent who holds a variance swap and hedges under the assumption of 
continuity, may super-replicate or sub-replicate the payoff depending on 
the form of the jumps. For example, for the kernel $H^L$ an agent who 
acts as if the price realisation can be assumed to be continuous will 
sub-replicate the variance swap if there are downward jumps and no upward 
jumps. Such an agent will underprice the swap.

We will use the analysis of this section to give us intuition about the 
extremal models which will lead to the price bounds on variance 
swaps derived in the Section~\ref{sec:bds}. 
The bounds will depend crucially on the kernel. Models under which the 
variance swap with kernel $H^L$ has highest price (assuming consistency 
with a given set of call prices) will be characterised by a single 
downward jump and no upward jumps. 

\begin{remark} \label{r:skew}
We will see later that the model which minimises the price for variance swaps with kernel $H^R$ also minimises the price for variation swaps with kernel $H^S$. If $\fpr$ has a quadratic variation, then in the continuous limit
$V_{H^S}(\fpr,P_\infty) = \sum_{0 < t \leq T} (\Delta \fpr(t))^3$. This payoff will be smallest if all jumps are downwards and we will see that if the call prices are given for expiry time $T$, then the model that produces the lowest price is one under which the price path has a single downward jump. 
\end{remark}

\subsection{A related Skorokhod embedding problem}
\label{sec:skorokhod}

In this section we relate the problem of finding extremal prices for the 
variance swap to a Skorokhod embedding problem. Again the aim is to 
develop intuition which will guide the derivation of the optimal 
model-free hedges in the next section.

Let $\mu$ be a measure on 
$\R^+$ with mean $m$ and
let $\sM_\mu$ be the class of all martingales such that for $X \in 
\sM_\mu$, $X_0=m$ and
$X_T \sim \mu$. For
each $X \in \sM_\mu$ there exists time-change $t \rightarrow A_t$, null 
at 0, such that $X_t=B_{A_t}$ for a Brownian motion $B$ started at $m$. 
Suppose that we have a filtered probability space 
$(\Omega, {\mathcal G}, {\mathbb G}, \Prob)$ such that $B$ is a 
${\mathbb G}$-Brownian motion with $B_0=m$. Then $X$ is adapted to the 
filtration ${\mathbb F}=({\mathcal F}_t)_{t \geq 0}$ where ${\mathcal 
F}_t = 
{\mathcal 
G}_{A_t}$.

Let $A^c$ be the continuous part of $A$. Note that 
$dA^c_t=(dX_t^c)^2=d[X]^c_t$.  Let $S^X=(S^X_t)_{t \geq 0}$ 
(respectively $S$) be the process of the running maximum of $X$ 
(respectively $B$) so that $S^X_t = \sup_{u \leq t} X_u$. Note that $X_t \leq 
S^X_t \leq S_{A_t}$ and then, 
path-by-path 
with $\Delta B_{A_t} = B_{A_t} - 
B_{A_{t-}}$, we have
\begin{eqnarray}
V_{H^R}(X,P_\infty) \;  =  \; \int_0^T \frac{d[X]^c_t}{(X_{t-})^2}
     + \sum_{0 \leq t \leq T} \left(\frac{\Delta 
X_t}{X_{t-}}\right)^2 
& \geq &  \int_0^T \frac{d[X]^c_t}{(S^X_{t-})^2} 
     + \sum_{0 \leq t \leq T} \left(\frac{\Delta X_t}{S^X_{t-}}\right)^2 
\label{eq:perkinslb1}\\ 
&\geq& \int_0^T \frac{d A_t^c}{(S_{A_{t-}})^2} + \sum_{0 \leq t \leq 
T} \left(\frac{\Delta B_{A_t}}{S_{A_{t-}}}\right)^2. \ \ 
 \label{eq:perkinslb2}
\end{eqnarray} 
We suppose, for the moment, that $\mu$ has a second moment. 
Then $(X_{t})_{0 \leq t \leq T}$ 
is a square-integrable martingale and we find that, \begin{eqnarray*} 
\E\left[\int_0^T \frac{dA_t^c}{(S_{A_{t-}})^2}+\sum_{0 \leq t \leq T} 
\left(\frac{\Delta B_{A_t}}{S_{A_{t-}}}\right)^2\right] &=& 
\E\left[\int_0^T \frac{dA^c_{t}+\Delta A_t}{(S_{A_{t-}})^2}\right] \\ 
&=& \E\left[\int_0^T \frac{dA_t}{(S_{A_{t-}})^2}\right] \\ &\geq& 
\E\left[\int_0^{A_T} \frac{du}{(S_u)^2}\right]. \end{eqnarray*} This 
motivates looking at the following problem: \begin{equation} \min_{\tau 
\in UI(\mu)} \E\left[\int_0^\tau \frac{du}{S_u^2}\right], \end{equation} 
where $UI(\mu)$ is the class of stopping times such that $B_\tau \sim 
\mu$ and $B_{t \wedge \tau}$ is uniformly integrable. This problem is a 
special case of a problem considered in Hobson and 
Klimmek~\cite{HobsonKlimmek:11a}, where it is proved that the minimum 
is attained by the Perkins embedding, which we will denote $\tau_\mu^P$.
Note that the Perkins solution of the Skorokhod embedding problem is 
generally defined for centred probability measures, but the translation
to measures with non-zero mean equal to the non-zero starting point is 
trivial.

Let $I = (I_t)_{t \geq 0}$ denote the infimum 
process $I_t = \inf_{u \leq t} B_u$.

\begin{theorem} \label{t:hobsonpedersen} [Perkins~\cite{Perkins:86}, Hobson and 
Pedersen~\cite{HobsonPedersen:02}]
Given $\nu$ a probability measure with support on $\R^+$, with mean $m$ 
let $Z_\nu$ denote a random variable with law $\nu$ and 
define 
$C_\nu(z) = \E[(Z_\nu - z)^+]$ and $P_\nu(z) = \E[(z - Z_\nu)^+]$.
Define also $\alpha=\alpha_{\nu}: (m,\infty) \mapsto [0,m)$ and 
$\beta=\beta_{\nu} : 
(0,m) 
\mapsto 
(m,\infty)$ 
by
\begin{equation} \label{eq:perkinsdef}
\alpha(z) = \arg \min_{y < m} \frac{C_{\nu}(z) - P_\nu(y)}{z-y},  
\hspace{10mm}
\beta(z) = \arg \min_{y > m} \frac{P_\nu(z) - C_\nu(y)}{y-z}.
\end{equation}

Let $B$ be Brownian motion started at $m$, with maximum process $S$ and 
minimum process $I$. Suppose $\mu$ has no atom at $m$. Then $\tau^P_\nu 
:= \inf \{ 
u> 0 : B_u < \alpha_\nu(S_u) \mbox{ or } B_u > \beta_\nu (I_u) \}$ 
solves the Skorokhod embedding problem for $\nu$ in the sense that 
$B_{\tau^P_\nu} \sim \nu$ and $(B_{t \wedge \tau^P_\nu})_{t \geq 0}$ is 
uniformly integrable.

If $\nu$ has an atom at $m$ then we assume ${\mathcal F}_0$ is 
sufficiently rich as to support a uniform random variable $\tilde{Z}_U$, 
which 
is 
independent of $B$. Then
\[ \tau^P_\nu := \left\{ \begin{array}{ll}
0 & \tilde{Z}_U \leq  \nu(\{m\}) \\
\inf \{ u> 0 : B_u < 
\alpha_\nu(S_u) 
\mbox{ or }
B_u >
\beta_\nu (I_u) \}  \hspace{5mm} \;& \tilde{Z}_U> \nu(\{m\}) \end{array} 
\right.
\]  
solves the Skorokhod embedding for $\nu$.
\end{theorem} 

The Perkins embedding has a minimality property in that for increasing 
functions $F$ it minimises $\E[F(S_\tau)]$ over embeddings $\tau$ of 
$\nu$. Moreover, as shown in \cite{HobsonKlimmek:11a} it also minimises 
the expected value of functionals of the joint law of the running 
maximum and terminal value $F(B_\tau,S_\tau)$ over stopping times $\tau$ 
in $UI(\nu)$, provided $F$ satisfies some consistency conditions. The 
salient characteristic of the Perkins embedding which results in 
optimality is that either $B_{\tau_\nu^P}=S_{\tau_\nu^P}$ or 
$B_{\tau_\nu^P}=\alpha_\nu(S_{\tau^P_\nu})$.

Now consider the problem of finding the consistent model for which 
$V_{H^R}(X,P_\infty)$ has lowest possible price, and recall that 
knowledge of call prices is equivalent to knowledge of the marginal law 
$\mu$ of $X_T$. 
To obtain the lowest possible price we
might expect
equality in each of (\ref{eq:perkinslb1})-(\ref{eq:perkinslb2}), and 
thus
that just before a jump, the process is at its current 
maximum. Moreover,  
the model 
should be related to the Perkins embedding. 
 
\begin{lemma} \label{l:qmartingale} Let $B$ be Brownian motion 
started at $m$. Let $\overline{H}_b = \inf \{ u \geq 0 : B_u = b \}$ be the 
first
hitting time of level $b$ by Brownian motion.
Let $\Lambda(t)$ be a strictly increasing, continuous function such that 
$\Lambda(0)=m$ and $\lim_{t \uparrow T} \Lambda(t)$ is infinite.

Define the process
$\tilde{Q}^\mu=(\tilde{Q}^\mu_t)_{0 \leq t \leq T}$ by
\begin{equation} \label{eq:qmartingale}
\tilde{Q}^\mu_t=B_{\Hb_{\Lambda(t)} \wedge \tau^P_\mu},
\end{equation}
and let ${Q}^\mu$ be the right-continuous modification of 
$\tilde{Q}^\mu$.

Then, ${Q}^{\mu}$ is a martingale such that ${Q}^{\mu}_T \sim \mu$. 
Moreover, the paths of $Q^\mu$ are continuous and increasing, except 
possibly at a single jump time. Finally, either $Q^{\mu}_T \equiv 
B_{\tau^P_\mu}=S_{\tau^P_\mu}$ or $Q^{\mu}_T \equiv 
B_{\tau^P_\mu}=\alpha_{\mu}(S_{\tau^P_\mu})$. 
\end{lemma}

\begin{proof} Since $\tau^P_\mu$ is finite almost surely
we have that $Q^{\mu}_T \equiv B_{\tau^P_\mu} \sim \mu$. Moreover, 
for 
$\Lambda(t) < \tau^P_\mu$, $Q^\mu_t = 
\Lambda(t)=B_{\Hb_\Lambda(t)}=S_{\Hb_{\Lambda(t)}}$. 
\end{proof}

The martingale $Q^\mu$ will be used in Section~\ref{sec:cts} to 
show that in the continuous-time limit, the bounds we obtain are tight. 
The martingale $Q^\mu$ is the related to the Perkins embedding in the 
same way that the Dubins-Gilat~\cite{DubinsGilat:1978} martingale is 
related to the Az\'ema-Yor~\cite{AzemaYor:79a} embedding.

We can also consider a reflected version of the martingale $Q^\mu$ based 
on the infimum process rather than the maximum process.

\begin{lemma} \label{l:rmartingale} 
Let $\lambda(t)$ be a strictly decreasing, continuous function such that
$\lambda(0)=m$ and $\lim_{t \uparrow T} \lambda(t)$ is zero.

Define the process
$\tilde{R}^\mu=(\tilde{R}^\mu_t)_{0 \leq t \leq T}$ by
\begin{equation} \label{eq:rmartingale}
\tilde{R}^\mu_t=B_{\Hb_{\lambda(t)} \wedge \tau^P_\mu},
\end{equation}
and let ${R}^\mu$ be the right-continuous modification of
$\tilde{R}^\mu$.

Then, ${R}^{\mu}$ is a martingale such that ${R}^{\mu}_T \sim \mu$.
Moreover, the paths of $R^\mu$ are continuous and decreasing, except
possibly at a single jump time. Finally, either $R^{\mu}_T \equiv
B_{\tau^P_\mu}=I_{\tau^P_\mu}$ or $R^{\mu}_T \equiv
B_{\tau^P_\mu}=\beta_{\mu}(I_{\tau^P_\mu})$.
\end{lemma}

\begin{remark} In this section we have exploited a connection between the 
problem of finding bounds on the prices of variance swaps and the Skorokhod 
embedding problem. This link is one of the recurring themes of the literature 
on the model-independent bounds, see Hobson~\cite{Hobson:10}. We exhibit this 
link for the kernel $H^R$, and in this sense at least, it seems that variance 
swaps defined via $H^R$ are the more natural mathematical object. Nonetheless, 
the intuition developed via $H^R$ and the Skorokhod embedding 
problem is valid more widely.
\end{remark}

\section{Path-wise Bounds for Variance Swaps} 
\label{sec:bds}

Previous sections have defined notation and developed intuition for the 
problem. Now we begin the construction of path-wise hedging strategies. 
We do this by defining a class of synthesisable payoffs with a useful 
extra property which can be exploited to give sub-hedges. Then, motivated 
by the results of Section~\ref{sec:skorokhod}, we define a further 
class of payoffs which are based on decreasing functions. Finally we 
show that for the kernel $H^R$, members of this new class belong to the 
former class also, and thus yield sub-hedges.

To construct a sub-hedge for a variation swap with kernel 
$H$ for any price realisation $\fpr$, suppose that there exists a 
pair of functions $(\psi,\delta)$ such that for $x,y \in \R$
\begin{equation} \label{eq:Hbound}
H(x,y) \geq \psi(y)-\psi(x)+\delta(x)(y-x).
\end{equation}
Then we may interpret $(\psi,\delta)$ as a semi-static hedging strategy 
(for a Markov and time-homogeneous dynamic strategy) and 
then
for any price realisation $\fpr$ and partition $P$, 
\[V_H(\fpr,P) \geq \psi(\fpr(T))-\psi(\fpr(0))-\sum_k 
\delta(\fpr(t_k))(\fpr(t_{k+1})-\fpr(t_k)).\]
By Definition~\ref{d:2.10subhedge} we have constructed a sub-hedge for the 
variation swap with kernel $H$.

Suppose now that $H$ is a variance swap kernel, and that $\psi$ is 
differentiable. Recall that $H_y(x,x)=0$.
Dividing both sides of (\ref{eq:Hbound}) 
by $y-x$ and letting $y \downarrow x$, we find that $\delta(x) \leq 
-\psi'(x)$. Similarly letting $y \uparrow x$, $\delta(x) \geq 
-\psi'(x)$. 
Thus if (\ref{eq:Hbound}) is to hold we must have that
$\delta \equiv -\psi'$ and our search for pairs of functions satisfying
(\ref{eq:Hbound}) is reduced to finding differentiable functions $\psi$ 
satisfying
\begin{equation} \label{eq:Hbound2}
H(x,y) \geq \psi(y)-\psi(x) - \psi'(x)(y-x).
\end{equation}
or equivalently,
$\psi(y) \leq H(x,y) + \psi(x) + \psi'(x)(y-x)$. Note that there is 
equality in this last expression at $y=x$.

\begin{definition} \label{d:psimin}
$\psi \in \Psi_0$ is a {\it candidate 
sub-hedge payoff} if for all 
$y 
\in (0,\infty)$,
\begin{equation} \label{eq:psimin}
\psi(y)=\inf_x \left\{H(x,y)+\psi'(x)(y-x)+\psi(x)\right\}.
\end{equation}
\end{definition}

Given a candidate sub-hedge payoff $\psi$ we can generate a candidate 
semi-static hedge $(\psi,\delta)$ by taking $\delta=-\psi'$. We will say that 
$\psi$ is the root of the semi-static sub-hedge $(\psi,- \psi')$.

It remains to show how to choose candidate sub-hedge payoffs 
and especially those which have good properties. Using the intuition 
developed in the previous section for the kernel $H^R$
we expect optimal sub-hedging strategies to be associated with the  
martingale $Q$ defined in (\ref{eq:qmartingale}).
For realisations of $Q$, either the path has no jump, or there is a 
single jump, and if the jump occurs when the process is at $x$ then the 
jump is to $\alpha(x)$.

With this in mind let ${\mathcal K}={\mathcal K}(\fpr(0))$ be the set of 
monotone decreasing right-continuous
functions $\kappa: [\fpr(0),\infty) \rightarrow (0,\fpr(0)]$, with 
$\kappa(\fpr(0))=\fpr(0)$. Let $k$ denote the inverse of $\kappa$.
For $y < f(0)$ we want the infimum in (\ref{eq:psimin}) 
to be attained at $x=k(y)$. 
Then $\psi$ must satisfy
\begin{equation}
\label{eq:psiz}
\psi(y)=H(k(y),y)+\psi(k(y))+\psi'(k(y))(y-k(y)).
\end{equation}     
Moreover, if $\psi'$ is differentiable, then for $x=k(y)$ to be the 
argument of the infimum in (\ref{eq:psimin}) we must have that $k$ 
satisfies $H_x(k(y),y)+\psi''(k(y))(y-k(y))=0$ or equivalently 
\begin{equation} \label{eq:secondderiv}
H_x(x,\kappa(x))=\psi''(x) (x-\kappa(x)).
\end{equation}
This suggests that we can define candidate sub-hedge payoffs
$\psi$ via (\ref{eq:secondderiv}) on $(f(0),\infty)$ and via
(\ref{eq:psiz}) on $(0,f(0))$.

If $\psi$ satisfies (\ref{eq:Hbound2}) then so does $\psi+a+b(y-x)$ for 
any $a$, $b$. Earlier we argued that without loss of generality for a 
semi-static hedging strategy we could assume $\psi'(f(0))=0$. Now we may 
restrict attention further to $\psi$ with $\psi(f(0))=0$.

Define $\Phi(u,y)=H_x(u,y)/(u-y)$. Write $\Phi^R(u,y)=H^R_x(u,y)/(u-y)$, and 
similarly for other kernels.

\begin{definition} \label{d:psi}
For $\kappa \in {\mathcal K}$ with inverse $k$, define 
$\psi_{\kappa,H} \equiv \psi_\kappa: 
(0,\infty) \mapsto \R^+$, by 
$\psi_\kappa(f(0))=0$ and
\[ \psi_\kappa = \left\{ \begin{array}{c}
\psi_\kappa(x) \\ \psi_\kappa(z) \end{array} \right\} = \left\{ 
\begin{array}{ll}
 \int_{f(0)}^x (x-u) \Phi(u,\kappa(u)) du  & x > f(0) \\
 \psi_\kappa(k(z)) + \psi_\kappa'(k(z)) (z-k(z)) + H(k(z),z) & z < f(0) 
\end{array}\right. \]
We call such a function a {\it candidate payoff of Class $\mathcal K$}. 
\end{definition} 
By convention we use the variable $x$ on $(f(0),\infty)$ and $z$ on 
$(0,f(0))$, to reflect the fact that $\psi$ is defined explicitly on the 
former set, but only implicitly on the latter. 

For the present we fix 
$\kappa$ and we write simply $\psi$ for $\psi_\kappa$. Note that the 
value of $\psi(x)$ does not depend on the right-continuity assumption 
for $\kappa$. Further, observe that if $\kappa$ is not injective and there is 
an interval $A_z \equiv \{x : \kappa(x)=z \} \subseteq (m,\infty)$
over which $\kappa$ takes the value $z$
then $k$ has a jump at $z$. Nonetheless, the 
value of $\psi(z)$ does not depend on the choice of $k(z)$. To see this, 
for $x \in A_z$ consider $\Psi(x) := \psi(x) 
+ \psi'(x)(z-x) +H(x,z)$. Then, on $A_z$, ${d \Psi}/{dx} = 
\psi''(x)(z-x) 
+ H_x(x,z) \equiv 0$, using (\ref{eq:secondderiv}).

Motivated by the results of Section 3.3 we have defined 
$\psi$ relative to the set of decreasing functions $\mathcal K$ with 
the aim of constructing a sub-hedge. However, there are 
analogous definitions based on constructing super-hedges or using the 
martingale $R$ or both.

\begin{definition} \label{d:psimax}
$\psi: (0,\infty) \rightarrow (0,\infty)$ is a {\it candidate
super-hedge payoff} if for all
$y
\in (0,\infty)$,
\begin{equation} \label{eq:psimax}
\psi(y)=\sup_x \left\{H(x,y)+\psi'(x)(y-x)+\psi(x)\right\}.
\end{equation}
\end{definition}

Define ${\mathcal L}={\mathcal L}(f(0))$ be the set of
monotone increasing functions $\ell: (0,f(0)) \rightarrow (f(0),\infty)$, with
$\ell(\fpr(0))=\fpr(0)$. Let $l$ be inverse to $\ell$.
\begin{definition}
For $\ell \in {\mathcal L}$ with inverse $l$, define $\psi_\ell:
(0,\infty) \mapsto
\R^+$, {\it the candidate payoff of Class $\mathcal L$}
by $\psi_\ell(f(0))=0$ and
\[ \psi_\ell = \left\{ \begin{array}{c}
\psi_\ell(x) \\ \psi_\ell(z) \end{array} \right\} = \left\{
\begin{array}{ll}
 \int^{f(0)}_x (u-x) \Phi(u,\ell(u)) du  & x < f(0) \\
 \psi_\ell(l(z)) + \psi_\ell'(l(z)) (z-l(z)) + H(l(z),z) & z > f(0)
\end{array}\right. \]
\end{definition}

Our next aim is to give conditions which guarantee that the semi-static 
strategy $(\psi,-\psi')$  
satisfies equation (\ref{eq:Hbound}). 
\begin{definition} \label{d:A1}
A variation swap kernel $H$ is an {\em increasing} (a {\em decreasing}) kernel 
if it 
is a regular variation swap kernel and
\begin{enumerate}
\item[(i)] $\Phi(u,y)$ is monotone increasing (decreasing) in 
$y$,
\item[(ii)] $H(a,b)+H_y(a,b)(c-b) \geq (\leq) H(a,c)-H(b,c)$ for all 
$a>b$.
\end{enumerate}
\end{definition}
The second condition in Definition~\ref{d:A1} is equivalent to the fact that
$H_{yy}(x,y)$ is increasing (decreasing) in its first argument.

\begin{example} \label{ex:indepkernel}
$H^R$ and $H^S$ are increasing kernels and $H^L$ is a decreasing kernel. The 
kernels $H^B$ and $H^Q$ are simultaneously both increasing and 
decreasing since $\Phi^B(u,y)=2 u^{-2}$ and 
$\Phi^Q(u,y)=2$ do not depend on $y$ and Condition (ii) in 
Definition~\ref{d:A1} is 
satisfied with equality in both cases.
\end{example}

\begin{example} \label{ex:gammakernel}
Consider the kernels $H^{G-}(u,y)=u H^R(u,y)$ and $H^{G+}(u,y)=y H^R(u,y)$. In 
the first case, variance is weighted by the pre-jump value of the price realisation 
and in the second case the variance is weighted by the post-jump value. 
Swaps of this type are known as Gamma swaps, 
see, for example, Carr and Lee \cite{CarrLee:10}. 
Both $H_{G-}$ and $H_{G+}$ are increasing kernels.
\end{example}

\begin{theorem} \label{t:mainthm}
\begin{enumerate}
\item[(i)]
\begin{enumerate}
\item[(a)]
If $H$ is an increasing kernel then every candidate payoff of 
Class ${\mathcal K}$ is the root of a semi-static sub-hedge for the 
kernel 
$H$.
\item[(b)]
If $H$ is an increasing kernel then every candidate payoff of 
Class ${\mathcal L}$ is the root of a semi-static super-hedge for the 
kernel 
$H$.
\end{enumerate}
\item[(ii)]
\begin{enumerate}
\item[(a)]
If $H$ is a decreasing kernel then every candidate payoff of
Class ${\mathcal L}$ is the root of a semi-static sub-hedge for the 
kernel
$H$.
\item[(b)]
If $H$ is an decreasing kernel then every candidate payoff of
Class ${\mathcal K}$ is the root of a semi-static super-hedge for the 
kernel
$H$.
\end{enumerate}
\end{enumerate}

\end{theorem}
\begin{proof} We will prove the theorem in the case (i)(a).
The proofs in the other cases are similar.

Fix $\kappa \in {\mathcal K}$ let 
$L_\kappa(x,y)=\psi_\kappa(x)+\psi'_\kappa(x)(y-x)+H(x,y)-\psi_\kappa(y)$. 
The result will follow if we can show that $L_\kappa(x,y) \geq 0$ for 
all $(x,y) \in (0,\infty)^2$. Since $\kappa$ is fixed we drop the 
subscript $\kappa$ in what follows.

Suppose that $x,z>f(0)$ and $y \in (0,\infty)$. Since 
$\psi(x)+\psi'(x)(y-x) = \int_{f(0)}^x (y-u) \Phi(u,\kappa(u)) du$ we 
have that 
\begin{eqnarray*}
L(x,y)-L(z,y) &=& \psi(x)+\psi'(x)(y-x)+H(x,y)-\psi(z)-\psi'(z)(y-z)-H(z,y) \\
&=&\int_z^x \left\{(y-u) \Phi(u,\kappa(u)) + H_x(u,y) 
\right\} du \\
&=& \int_z^x \left\{\Phi(u,y)- 
\Phi(u,\kappa(u)) \right\} (u-y) du . 
\end{eqnarray*}
If $y \geq f(0)$, then set $z=y$ to find that 
\[L(x,y)=\int_y^x \left\{\Phi(u,y) - 
\Phi(u,\kappa(u)) \right\} (u-y) du.\]
Since $y \geq f(0) \geq \kappa(u)$, $\Phi(u,y) \geq 
\Phi(u,\kappa(u))$ for all $u$. Hence 
$L(x,y) \geq 0$ with equality at $y=x$.

If $y<f(0)$ and $k$ is continuous at $y$ set $z=k(y)$. Otherwise, for 
definiteness set $z=k(y+)$. Then $L(k(y+),y)=0$ and
\[L(x,y)=\int_{k(y+)}^x \left\{ \Phi(u,y) - 
\Phi(u,\kappa(u)) \right\} (u-y) du.\]

If $k(y+) \leq x$ then $y \geq \hat{x}$, for all $\hat{x} \in 
[\kappa(x+),\kappa(x-)]$. Then for $u \in (k(y+),x)$, $\kappa(u) \leq y$ and 
since $\Phi(u,z)$ is increasing in $z$, the integrand is 
positive.

If $x < k(y+)$, then $y<\hat{x}$ for all $\hat{x} \in 
[\kappa(x+),\kappa(x-)]$. Then for $u \in (x,k(y+))$ we have $\kappa(u) 
> y$. Then again $L(x,y) \geq 0$.

Finally, we show that $L(x,y) \geq 0$ when $x<f(0)$. Note that since, by 
what we have shown above, $L(k(x),y) \geq 0$ it will suffice to show 
that $L(x,y) \geq L(k(x),y)$. But,
\begin{eqnarray*}
L(x,y)-L(k(x),y) &=& \psi(x)+\psi '(x)(y-x) + H(x,y) \\
&& -\psi(k(x))-\psi'(k(x))(y-k(x))-H(k(x),y) \\
&=& \psi(k(x))+\psi'(k(x))(x-k(x))+H(k(x),x)+\psi'(k(x))(y-x) \\&&+H_y(k(x),x)(y-x) +H(x,y)-\psi(k(x))-\psi'(k(x))(y-k(x))-H(k(x),y) \\
&=& H(k(x),x)+H(x,y)+H_y(k(x),x)(y-x)-H(k(x),y) \\
&\geq& 0,
\end{eqnarray*}
where the last inequality follows from Definition (\ref{d:A1}). 
\end{proof}

\section{The most expensive sub-hedge}
\label{sec:expensive}

In the next three sections we concentrate on lower bounds and increasing 
variance kernels, but there are equivalent results for upper bounds 
and/or decreasing variance kernels.

In this section we fix the call prices and attempt to identify the most 
expensive sub-hedge from the set of sub-hedges generated by candidate 
payoffs of Class ${\mathcal K}$. The price of this sub-hedge provides a 
highest model-independent lower bound on the price of the variance swap 
in a sense which we will explain in the section on continuous limits.

Associated with the set of call prices $C(k)$ (and put prices 
$C(k)+f(0)-k$ given by put-call parity) there is a measure $\mu$ on 
$\R^+$ with mean $m$. 
Since $f$ is a forward price we must 
have $f(0)=m$.
Write $C=C_\mu$ to emphasise the connection 
between these quantities. Then $C(k)=C_\mu(k)=\int_k^\infty (x-k) 
\mu(dx)$. Recall that $C_\mu$ is convex so that $\mu(dx) = C_\mu''(x)dx$ 
with the right-hand-side to be interpreted in a distributional sense as 
necessary. We wish to calculate the cost of the European claim which 
forms part of the semi-static sub-hedge. By construction this 
is equal to 
$\int_{\R^+} \psi(x) \mu(dx)= \int_0^m \psi''(z) (C_\mu(z)+m-z) dz + 
\int_m^\infty \psi''(x)C_\mu(x)dx$.

\begin{proposition} For $H$ a variance swap kernel and 
$\kappa \in {\mathcal K}(m)$,
\begin{equation} \label{eq:5.1}
\int_0^\infty \psi_\kappa(x) \mu(dx)  = 
\int_0^m \mu(dz) H(m,z) + \int_m^\infty du \Sigma^{(u)}_{\mu}(\kappa(u))
\end{equation}
where, for $v < m < u$,
\[ \Sigma^{(u)}_{\mu}(v) =
\Phi(u,v)
C_\mu(u) +
\int_{(0,v]}  \mu(dz) (u-z)
\left\{ \Phi(u,z) - \Phi(u,v) \right\}. \]
\end{proposition}

\begin{proof}
Let $\psi=\psi_\kappa$. Note that by definition $\psi(m)=0$, so there is no contribution from mass at $m$ and we can divide the integral on the left of (\ref{eq:5.1}) into intervals 
$(0,m)$ and $(m,\infty)$. 
For the latter,
\begin{eqnarray*}
\int_m^\infty \psi(x) \mu(dx) & = & \int_m^\infty  \mu(dx)  
\int_m^x (x-u) \Phi(u,\kappa(u)) du \\
& = &  \int_{u=m}^\infty du \Phi(u,\kappa(u)) \int_u^\infty (x-u) 
\mu(dx) \\
& = & \int_{u=m}^\infty du \Phi(u,\kappa(u)) C_{\mu}(u) =: I_1 .
\end{eqnarray*}

Now consider $\int_0^m \psi(z) \mu(dz)$. For this, using $H(k,z) = 
H(m,z) + \int_m^k H_x(u,z) du$ and
$\psi(x) + \psi'(x)(z-x) = \int_m^x du (z-u) \Phi(u,\kappa(u))$ 
we have
\begin{eqnarray*} 
\int_0^m \psi(z) \mu(dz) 
& = & \int_0^m \mu(dz) H(m,z) + \int_0^m \mu(dz)\int_m^{k(z)} du (u-z) \left\{
\Phi(u,z)- \Phi(u,\kappa(u)) \right\} \\
& =: & I_2 + I_3
\end{eqnarray*}
Note that $I_2$ depends on $H$ but not on $\kappa$.
Moreover, $I_3$ does not depend on the particular values 
chosen for the inverse 
taken over intervals of constancy of 
$\kappa$. (If $x<\tilde{x}$ are a pair of possible values for $k(z)$ then 
$\int_x^{\tilde{x}} du (u-z) \{\Phi(u,z) - \Phi(u, \kappa(u)) \} = 0$
since over this range $\kappa(u)=z$.)
Changing the order of integration we have
\[ I_3  =  \int_m^\infty du \int_{(0,\kappa(u)]} \mu(dz) 
(u-z)  \left\{
\Phi(u,z) - \Phi(u, \kappa(u)) \right\}, \]
and then
\( 
I_1 + I_3 
= \int_m^\infty du \Sigma^{(u)}_{\mu}(\kappa(u)).
\)
\end{proof}

Our goal is
to maximise the expression 
(\ref{eq:5.1})
over decreasing functions $\kappa \in 
{\mathcal K}$. 
As noted above, $I_2$ is independent of $\kappa$, and to
maximise $\int_m^\infty du \Sigma^{(u)}_\mu(\kappa(u))$ we can maximise
$\Sigma^{(u)}_\mu(\kappa)$
separately for each $u>m$, and then check that the minimiser is a 
decreasing function of $u$.

\begin{proposition} \label{p:bestkappa} Suppose $H$ is an increasing variance swap kernel. Then
$\int_0^\infty \psi_\kappa(x) \mu(dx)$ is maximised over 
$\kappa \in {\mathcal K}$ by $\kappa = \alpha$ where $\alpha$ is the 
quantity 
which arises in (\ref{eq:perkinsdef}) in the definition of the Perkins 
solution to the Skorokhod embedding problem.
\end{proposition} 

\begin{proof}
For $u>m$ consider 
$\Theta^{(u)}_\mu(v):= C_{\mu}(v) - \int_{(0,v]} \mu(dz) (u-z)$ defined 
for $v \in (0,u)$. Then for each $u$, $\Theta^{(u)}_\mu$ is a strictly decreasing 
right-continuous function taking both positive and negative values on $(0,m)$. 
Let $\overline{\kappa} = \overline{\kappa}(u) = \sup \{ v 
: \Theta^{(u)}_\mu(v) \geq 0 \}$.
We have
$\Theta^{(u)}_\mu(\overline{\kappa}-) \geq 0 \geq 
\Theta^{(u)}_\mu(\overline{\kappa}+)$.

Suppose $H$ is an increasing variance swap kernel so that $\Phi(u,y)$ 
is increasing in $y$. We want to show that
$\Sigma^{(u)}_\mu(v)$ is maximised by 
$v=\overline{\kappa}(u)$.

Suppose $m> v > \overline{\kappa}(u)$. We aim to show that for all
$\kappa \in (\overline{\kappa}(u),v)$
we have $\Sigma^{(u)}_\mu(v) \leq \Sigma^{(u)}_\mu(\kappa)$. We have 
\begin{eqnarray}
\Sigma^{(u)}_\mu(v) - \Sigma_\mu^{(u)}(\kappa) & = &
\Phi(u,v) C_{\mu}(u) +
\int_0^{v}  \mu(dz) (u-z)
\left\{ \Phi(u,z) - \Phi(u,v) \right\} \nonumber \\
&& \hspace{5mm}
- \Phi(u,\kappa) C_{\mu}(u) -
\int_0^{\kappa}  \mu(dz) (u-z)
\left\{ \Phi(u,z) - \Phi(u,\kappa) \right\} \nonumber \\
& = & \int_{\kappa}^{v} \mu(dz) (u-z) \left\{ \Phi(u,z) - 
\Phi(u,v)  
\right\} + \left[ \Phi(u,v) - \Phi(u,\kappa) \right] 
\Theta^{(u)}_\mu(\kappa).
\nonumber 
\end{eqnarray}
Since $H$ is an increasing variance kernel, for $z \in (\kappa, v)$,
$\Phi(u,z) \leq \Phi(u,v)$, and the first integral is non-positive. 
Furthermore, $ \Phi(u,v) \geq \Phi(u,\kappa)$ and 
$\Theta^{(u)}(\kappa)<0$. Hence we conclude that
$\Sigma^{(u)}_\mu(v) \leq \Sigma_\mu^{(u)}(\kappa)$. 

Similar arguments 
show that if $v < \overline{\kappa}(u)$ then 
$\Sigma^{(u)}_\mu(v) \leq \Sigma_\mu^{(u)}(\kappa)$ for any $\kappa 
\in
(v, \overline{\kappa}(u))$, and
it follows that $\kappa =\overline{\kappa}(u)$ is a maximiser of
$\Sigma^{(u)}_\mu(v)$. 
 
Note that $\overline{\kappa}(u)$ is precisely the quantity 
$\alpha$ which arises in the Perkins construction. Hence 
$\overline{\kappa}$ is a decreasing function. Moreover, the definition
$\overline{\kappa}(u) = \sup \{ v
: \Theta^{(u)}_\mu(v) \geq 0 \}$ ensures that $\overline{\kappa}$ is 
right continuous.
\end{proof}

\begin{corollary} \label{c:kappapprox}
Suppose $\kappa_n(x)$ is a sequence of elements of ${\mathcal K}$ with
$\kappa_n(x) \downarrow \overline{\kappa}(x)$. Then 
$\int_{[0,\infty)} \psi_{\kappa_n}(x) \mu(dx)$ converges monotonically 
to $\int_{[0,\infty)} \psi_{\overline{\kappa}}(x) \mu(dx)$.
\end{corollary}

\begin{proof}
Recall that $\int_{[0,\infty)} \psi_\kappa(x) \mu(dx) = \int_0^1 \mu(dz) 
H(1,z) + \int_1^\infty du \Sigma^{(u)}_\mu({\kappa}(u))$. By 
the above 
arguments we have that  $\Sigma^{(u)}_\mu(z)$ is increasing in $z$ for 
$z>\overline{\kappa}(u)$. Hence the result follows by monotone 
convergence.
\end{proof}

\begin{example} 
Let $H=H^R$, an increasing variance kernel. Let $\mu=U[0,2]$ and let 
$\kappa:[1,2]\rightarrow [0,1]$ be given by 
$\kappa(x)=\alpha_\mu(x)=x-2\sqrt{x-1}$. Similarly we define 
$\ell(x)=\beta_\mu(x)=x+2\sqrt{1-x}$. Then $(\psi_{\kappa},-\psi_{\kappa}')$ 
is the most expensive sub-hedge of class ${\mathcal K}$ and 
$(\psi_\ell,-\psi_\ell ')$ is the cheapest super-hedge of class ${\mathcal 
L}$. Although we cannot calculate the functions $\psi_\kappa, \psi_\ell$ 
explicitly, they can be 
evaluated numerically, see the left hand side of Figure 2. Now suppose 
$H=H^L$. The roles of $\psi_\kappa$ and $\psi_\ell$ are reversed (see the 
right hand side of Figure 2) and $(\psi_\kappa,- \psi'_\kappa)$ is the 
root 
of a semi-static
super-hedge and 
$(\psi_\ell, -\psi'_\ell)$ is the root of a semi-static sub-hedge.

\begin{figure}[htb]
\begin{center}
$\begin{array}{c@{\hspace{1in}}c}
\includegraphics[height=5cm,width=6cm]{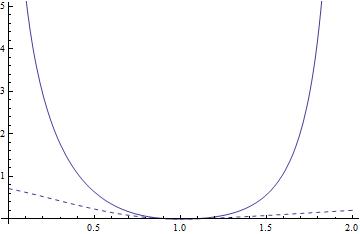} &
\includegraphics[height=5cm,width=6cm]{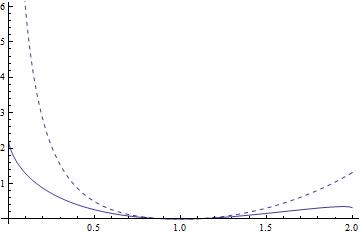} \\
\end{array}$
\end{center}
\caption{For the two kernels $\psi_{\kappa}$ is shown as a 
dashed line and $\psi_{\ell}$ is shown as a solid line. For the kernel $H^R$ 
(left-hand-side), 
$\psi_\kappa$ is associated with a lower bound on the price of the variance 
swap. For the kernel $H^L$ (right-hand-side) $\psi_\kappa$ is associated with 
an upper bound. }
\end{figure}
\end{example}

\section{Continuous limits and the tightness of the bound} 
\label{sec:cts}

The bounds we have constructed based on the functions $\psi_\kappa$ hold 
simultaneously across all paths and all partitions. The purpose of this 
section is to consider the limit as the partition becomes finer. It will turn 
out that in the continuous limit there is a stochastic model which is 
consistent with the observed call prices and for which there is equality in 
the inequality (\ref{eq:Hbound}) from which we derive the lower bound. In this 
sense the model-free bound is optimal, and can be attained.

The analysis of this section justifies restricting attention to 
candidate payoffs of Classes ${\mathcal K}$ and $\mathcal L$.  
Hedges of this type either sub-replicate or super-replicate the payoff 
of the variance swap depending 
on the 
form of the kernel, but there could be other sub- and super-replicating 
strategies which do not take this form. In principle, for a given 
partition one of these other sub-hedges could give a tighter model-independent 
bound than we can derive from our analysis. (As an extreme example, suppose 
the 
partition is trivial ($0=t_0 < t_1=T$). Then $V_H(f,P)= H(f(0),f(T))$ which 
can be 
replicated exactly using call options.) However, in the continuous limit our 
bound is best possible, so that when the partition is finite, but the mesh 
size is small we expect our hedge to be 
close to best possible and
relatively simple to implement.

For a finite partition $P^{(n)}$ in the dense sequence $\pseq = 
(P^{(n)})_{n \geq 1}$ we have
\begin{equation} \label{eq:1a}
V_H(f,P^{(n)}) = \sum_{k=0}^{N^{(n)}-1} H(f(t_k),f(t_{k+1})) \geq 
\psi(f(T)) - \psi(f(0)) - \sum_{k=0}^{N^{(n)}-1}
\psi'(f(t_k)) ( f(t_{k+1})-f(t_k)). 
\end{equation}
We want to conclude that 
the limits 
$V_H(f,P_\infty) = \lim_n V_H(f,P^{(n)})$ and
\begin{equation} \label{eq:limitofsums}
\lim_n \sum_{k=0}^{N^{(n)}-1} \psi'(f(t_k)) ( f(t_{k+1})-f(t_k)) =
\int_0^T \psi'(f(t-))df(t)
\end{equation}
exist for each path under consideration. Our analysis follows the development 
of a path-wise It\^o's formula in 
F\"{o}llmer~\cite{Follmer:79}. Let $\epsilon_t$ denote a point mass at $t$.

\begin{definition} \label{d:fquad}
A path realisation $\fpr$ has a quadratic variation on a dense sequence 
of 
partitions $\pseq=(P^{(n)})_{n \geq 1}$ if, when we define the measure
\[\zeta_n = \sum^{N^{(n)}-1}_{{k=0}, \ {t_k \in P^{(n)}}} 
(f(t_{k+1})-f(t_k))^2 \epsilon_{t_k},\]
then 
the sequence $\zeta_n$ converges weakly to a Radon measure $\zeta$ on 
$[0,T]$. Then $([f]_t)_{t \geq 0}$ is given by $[f]_t = \zeta([0,t])$. 
\end{definition}

The atomic part of $\zeta$ is given by squared jumps of $f$. 
Moreover the quadratic variation $([f]_t)_{t \geq 0}$ is simply the 
cumulative mass function of $\zeta$.

\begin{theorem}~(F\"{o}llmer~\cite{Follmer:79}) \label{t:follmer}
Suppose the price realisation $\fpr$ has a quadratic variation along 
$\pseq=(P^{(n)})_{n \geq 1}$ and $G$ is a twice continuously differentiable 
function from $\R^+$ to $\R$, then 
\[\int_0^T G'(f(t-)) df(t)=\lim_{n \uparrow \infty} 
\sum_{t=0}^{N^{(n)}-1} G'(f(t_k))(f(t_{k+1})-f(t_k))\] 
exists and 
\begin{eqnarray*}
G(f(T))-G(f(0)) &=& \int_0^T G'(f(s-))df(s) + \frac{1}{2} \int_{(0,T]} 
                         G''(f(s))d[f]^c_s \\ 
&&+ \sum_{s \leq T} \left[ G(f(s))-G(f(s-))-G'(f(s-))\Delta f(s) \right],
\end{eqnarray*}
and the series of jump terms is absolutely convergent. 
\end{theorem}
Hence, provided $\psi$ is twice continuously differentiable on the support of 
$f$ and $f$ has a quadratic variation along $\mathcal P$, it follows 
immediately that the 
limit in (\ref{eq:limitofsums}) exists. In our setting 
$\psi''_\kappa(u)=\Phi(u,\kappa(u))$ for $u>1$, so 
that a sufficient condition for
$\psi''_\kappa(u)$ to be continuous on $(1,\infty)$ is that
$\kappa$ is continuous.
Further, on $u<1$, provided $k\equiv \kappa^{-1}$ is differentiable and 
$H_y$ exists, we have $\psi'(z)=\psi'(k(z))+H_y(k(z),z)$. Hence, 
sufficient conditions for $\psi$ to be twice continuously differentiable 
on $(0,1)$ are that $k$ is continuously differentiable, $\kappa$ is 
continuous and $H_{xy}$ and $H_{yy}$ are continuous.
Let $\bbarkap$ be the class of decreasing functions $\kappa: 
(f(0),\infty) \rightarrow (0,f(0))$ which are continuous and have an 
inverse $k$ which is continuously differentiable. 

\begin{corollary} Suppose that $H$ is an increasing variance kernel, and that 
$f$ has a quadratic variation.
Suppose $\kappa \in \bbarkap$ and $\psi = \psi_\kappa$. Then the limit 
in (\ref{eq:limitofsums}) exists. 
\end{corollary}
Now we want to consider $V_H(f,P_\infty)=\lim_n V_H(f,P^{(n)})$. 
\begin{lemma} \label{l:d}
Suppose $H$ is a variance swap kernel. If 
$\pseq=(P^{(n)})_{n \geq 1}$ is a dense sequence of partitions, and
$\fpr$ has a quadratic variation 
along $\mathcal P$, 
then $\lim_{n \uparrow \infty} V_H(f,P^{(n)})$ exists and satisfies 
\begin{equation} \label{eq:4.0}
V_H(f,P_{\infty}) = \int_{(0,T]} \frac{1}{f(t-)^2} d[f]_t + \sum_{0 < t \leq 
T} 
H(f(t-),f(t)) - \sum_{0 < t \leq T} \frac{1}{f(t-)^2} (\Delta f(t))^2.
\end{equation} 
\end{lemma}

\begin{proof} Our proof follows F\"{o}llmer~\cite{Follmer:79}.
Fix $\epsilon > 0$. Partition $[0,T]$ into two classes: 
a finite class $C_1=C_1(\epsilon)$ of jump times 
and a class $C_2=C_2(\epsilon)$ such that 
\begin{equation} \label{eq:jumpepsilon}
\sum_{s \in [0,T], \ s \in C_2(\epsilon)} (\Delta f(s))^2 \leq \epsilon^2.
\end{equation}
Then $\sum_{k=0}^{N^{(n)}-1} H(f(t_k),f(t_{k+1})) = \sum_1 
H(f(t_k),f(t_{k+1})) 
+ \sum_2 H(f(t_k),f(t_{k+1}))$, where $\sum_1$ indicates a sum over those 
$0 \leq k \leq N^{(n)} -1$ for which $(t_k,t_{k+1}]$ contains a jump of class 
$C_1$. It follows that 
\begin{equation} \label{eq:Hsumjumps}
\lim_{n \uparrow \infty} \sum_1 H(f(t_k),f(t_{k+1})) = \sum_{t \in C_1(\epsilon)} H(f(t-),f(t)).
\end{equation}
On the other hand, using the properties $H(x,x)=0$, $H_y(x,x)=0$ we have from 
Taylor's formula that $H(x,y) = \frac{1}{2} H_{yy}(x,x)(y-x)^2+r(x,y)$. Using 
the fact that $(f(t))_{0 \leq t \leq T}$ is a compact subset of $(0,\infty)$ 
we may assume that the remainder term satisfies $|r(x,y)| \leq 
R(|y-x|)(y-x)^2$ where $R$ is an increasing function on $[0,\infty)$ such 
that $R(c) \rightarrow 0$ as $c \rightarrow 0$. Then 
\begin{eqnarray} \label{eq:H2sumexp}
\sum_2 H(f(t_k),f(t_{k+1})) &=& \frac{1}{2} \sum_2 
H_{yy}(\fpr(t_k),\fpr(t_k))(\fpr(t_{k+1})-\fpr(t_k))^2 
+ \sum_2 r(f(t_k),f(t_{k+1})) \nonumber \\
&=& \frac{1}{2} \sum H_{yy}(f(t_k),f(t_k))(f(t_{k+1})-f(t_k))^2 \nonumber \\
&& \; - \frac{1}{2} \sum_1 H_{yy}(f(t_k),f(t_k))(f(t_{k+1})-f(t_k))^2 
\nonumber \\
&& \; + \sum_2 r(f(t_k),f(t_{k+1})).
\end{eqnarray}

Since $H_{yy}(f,f)=2/f^2$ is uniformly continuous over the bounded set of 
values $(f(t))_{0 \leq t \leq T}$, by $(9)$ in F\"{o}llmer~\cite{Follmer:79}, 
the first term in (\ref{eq:H2sumexp}) converges to $\int_{(0,T]} 
\frac{1}{f(t-)^2} d[f]_t$ and the second term converges to $-\sum_{s \in C_1} 
\frac{1}{f(t-)^2} (\Delta f(t))^2$. Using (\ref{eq:jumpepsilon}) and the fact 
that the remainder term satisfies $|r(x,y)| \leq R(|y-x|)(y-x)^2$ we have 
that the last term is bounded by $R(\epsilon) [f]_T$. Finally, letting 
$\epsilon \downarrow 0$ we conclude that $V_H(f,P_\infty)=\lim_n 
V_H(f,P^{(n)})$ exists 
and (\ref{eq:4.0}) follows.
\end{proof}

\begin{corollary} \label{cor:ctslimit}
$V_{H^R}(f, P_{\infty}) = \int_{(0,T]} f(t-)^{-2} d[f]_t$ and
$V_{H^L}(f, P_{\infty}) = [\log f]_T$.
\end{corollary}

Combining (\ref{eq:1a}) with Theorem \ref{t:follmer} and Lemma \ref{l:d} it 
follows that for a path of finite quadratic variation and $\psi$ a 
twice-continuously differentiable function with $\psi(f(0))=0$, 
\begin{equation} \label{l:4.1}
V_H(f, P_{\infty} ) \geq \psi(f(T))-\int_0^T \psi'(f(t-))df(t).
\end{equation}
The left hand side is the payoff of the variance swap in the continuous limit. 
The expression on the right can be interpreted as the payoff of a semi-static 
hedging strategy $(\psi,-\psi')$ under continuous trading. From 
Definition \ref{d:2.9semistatic} for 
each of the partitions in the sequence we have that the price of the 
semi-static hedge is 
\begin{equation}
\label{eq:ctsprice}
\int_0^\infty \psi(x) \mu(dx)= \int_{f(0)}^\infty \psi''(x) 
C_\mu(x)dx+ 
\int_0^{f(0)} \psi''(z) (C_\mu(z) + f(0) - z) dz.
\end{equation} 
Since this value does not depend on the 
partition, in the
continuous-time setting we define the price of sub-hedge $(\psi,-\psi')$ 
to also be the expression given in (\ref{eq:ctsprice}).

\begin{corollary} \label{c:corF}

Suppose $H$ is an increasing variance swap kernel. A model-independent lower 
bound on the price of the continuous time limit of the variance swap with 
payoff $V_H(f)$ is  
\begin{equation} \label{eq:4.2}
\sup_\kappa \int_0^\infty \psi_\kappa(x) \mu(dx) = \int_0^\infty \psi_{\alpha_\mu}(x) \mu(dx)
\end{equation}
where $\alpha_\mu$ is the quantity arises in the Perkins embedding (Theorem \ref{t:hobsonpedersen}).
\end{corollary}

\begin{proof} For any decreasing function $\kappa \in \bbarkap$ we can 
construct $\psi_\kappa$ such that $\int_0^\infty \psi_\kappa(x) \mu(dx)$ 
is the price of a sub-hedge for $V_H$ for any partition, and this 
continues to hold in the continuous-time limit. Moreover, by optimising 
over $\kappa$ we obtain a bound $\int_0^\infty \psi_{\alpha_\mu}(x) 
\mu(dx)$ which is the best bound of this form by Proposition 
\ref{p:bestkappa}. Note that even if $\alpha_\mu$ is not in class 
$\bbarkap$, by Corollary \ref{c:kappapprox} we can approximate it from 
above by a sequence of elements of class $\bbarkap$ such that in the 
limit we obtain the price $\int_0^\infty \psi_{\alpha_\mu}(x) \mu(dx)$ 
as a bound. \end{proof}

Our goal now is to show that this is a best bound in general and not just an 
optimal bound based on inequalities such as (\ref{eq:1a}) for $\psi \equiv 
\psi_\kappa$ and $\kappa$ a decreasing function. We do this by showing that 
there is a consistent model for which the price of the continuously monitored 
variance swap is equal 
to $\int_0^\infty \psi_{\alpha_\mu}(x) \mu(dx)$.

\begin{theorem} \label{t:theoremG}
There exists a consistent model such that 
\begin{equation} \label{eq:6.1pathwise}
V_H((X_t)_{0 \leq t \leq T},P_{\infty})  = 
\psi_{\alpha_\mu}(X_T)-\int_0^T \psi'_{\alpha_\mu}(X_{s-}) 
dX_s.
\end{equation}
\end{theorem}

\begin{proof} 
Recall Definition~\ref{d2.15} and note that we are given a set of call prices
and that in constructing a consistent model we are free to design an 
appropriate probability space $(\Omega, \mathcal{F}, 
\mathbb{F}=(\mathcal{F}_t)_{0 \leq t \leq T},
\Prob)$ as well as a stochastic process $(X_t)_{t \geq 0}$. 

Suppose we are given call prices $C(x)=C_\mu(x)$ for some $\mu$. 
Let $(\Omega, \mathcal{G}, \mathbb{G}=(\mathcal{G}_t)_{0 \leq t \leq T}, 
\Prob)$ support a Brownian motion $(W_u)_{u \geq 0}$ with initial value 
$W_0 =f(0)=\int_{\R^+} x \mu(dx)$ and suppose $\mathcal{G}_0$ contains a 
$U[0,1]$ 
random variable which is independent of $W$. (This last condition is necessary 
purely to ensure that the Perkins embedding of $\mu$ can be defined when $\mu$ 
has an 
atom at $f(0)$. If $\mu$ has no atom at $f(0)$ then we may take 
$\mathcal{G}_0$ to be trivial.)

Let $\tau_\mu^P$ be the Perkins embedding of $\mu$ in $W$. Write $S$ for the 
maximum process of $W$ so that $S_u=\max_{v \leq u} W_v$. Write $\Hb_x$ for 
the first hitting time by $W$ of $x$. Let $(\Lambda(t))_{0 \leq t \leq T}$ be 
a strictly increasing continuous function with $\Lambda(0)=\fpr(0)$ and 
$\lim_{t \uparrow 
T} \Lambda(t) = \infty$. Now define the left-continuous process 
$\tilde{X}=(\tilde{X}_t)_{0 
\leq t \leq T}$ via
\[ \tilde{X}_t = \left\{ \begin{array}{ll}
\Lambda(t) \hspace{10mm} & \mbox{$\Hb_{\Lambda(t)} \leq \tau_\mu^P$} \\
W_{\tau_\mu^P}  \; & \mbox{$\tau_\mu^P < \Hb_{\Lambda(t)}$.}
\end{array} \right. \]

Note that the condition $\Hb_{\Lambda(t)} \leq \tau_\mu^P$ can be 
re-written as 
$\Lambda(t) \leq S_{\tau_\mu^P}$ or equivalently $t \leq 
\Lambda^{-1}(S_{\tau_\mu^P})$. Define also 
$\tilde{\mathcal{F}}_t=\mathcal{G}_{\bar{H}_{\Lambda(t)}}$. 
Then 
$\tilde{X}$ 
is adapted to the filtration 
$\tilde{\mathbb{F}}=(\tilde{\mathcal{F}}_t)_{0 \leq 
t \leq T}$ 
and $\tilde{X}$ is a $\tilde{\mathcal{F}}$-martingale for which 
$\tilde{X}_T = 
W_{\tau_\mu^P} \sim 
\mu$.

In order to construct a right-continuous martingale with the same 
properties, for $t<T$ we set ${\mathcal F}_t = \cap_{u>t} \tilde{F}_t$ 
and $X_t = \lim_{u \downarrow t} \tilde{X}_u$, and for $t=T$ we set 
${\mathcal F}_T = \tilde{F}_T$ $X_T=\tilde{X}_T$.
Then $X$ is a right-continuous $\mathcal F$ martingale such that 
$(\Omega, 
\mathcal{F}, \mathbb{F}=(\mathcal{F}_t)_{0 
\leq 
t \leq 
T}, \Prob)$ is a consistent model.

Now we want to show that for this model (\ref{eq:6.1pathwise}) holds path-wise. 
Writing $\psi$ for $\psi_{\alpha_\mu}$, and $X_t$ as shorthand for each 
$X_t(\omega)$ we have for each $\omega$ 
\begin{eqnarray*}
\psi(X_T)- \int_0^T \psi'(X_{t-}) dX_t &=& \psi(W_{\tau^P_\mu}) - 
\int_{t=0}^{{\Lambda^{-1}}(S_{\tau_\mu^P})} \psi'(\Lambda(t)) d\Lambda(t)- \psi'(S_{\tau_\mu^P})(W_{\tau_\mu^P}-S_{\tau_\mu^P}) \\
&=& \psi(W_{\tau^P_\mu}) - \int_{f(0)}^{S_{\tau_\mu^P}} \psi'(u) du- 
\psi'(S_{\tau_\mu^P})(W_{\tau_\mu^P}-S_{\tau_\mu^P}) \\
&=& \psi(W_{\tau^P_\mu}) - \psi(S_{\tau_\mu^P})- 
\psi'(S_{\tau_\mu^P})(W_{\tau_\mu^P}-S_{\tau_\mu^P}).
\end{eqnarray*}

There are two cases. Either $W_{\tau_\mu^P}=S_{\tau_\mu^P}$, in which case 
this expression is equal to $0$ or, 
$W_{\tau_\mu^P}=\alpha_\mu(S_{\tau_\mu^P})$ and then the expression becomes 
\[\psi(\alpha_\mu(s))-\psi(s)-\psi'(s)(\alpha_\mu(s)-s) \equiv 
H(s,\alpha(s))\] at $s=S_{\tau_\mu^P}$, using 
Definition~\ref{d:psi}. In either case the right hand side of 
(\ref{eq:6.1pathwise}) is $H(S_{\tau_\mu^P},W_{\tau_\mu^P})$. For the left 
hand side of (\ref{eq:6.1pathwise}), $[X]_T^c=0$ and $(\Delta 
X_u)^2=(S_{\tau_\mu^P}-W_{\tau_\mu^P})^2 1_{ \{u = 
\Lambda^{-1}(S_{\tau_\mu^P}) \} } 
1_{ \{ W_{\tau_\mu^P} \neq S_{\tau_\mu^P} \}}$ so that from 
(\ref{eq:4.0}), 
$V_H(f,P_{\infty})=H(S_{\tau_\mu^P},W_{\tau_\mu^P})$. Hence 
(\ref{eq:6.1pathwise}) holds 
path-wise.

\end{proof}

\begin{corollary} \label{c:H}
Suppose $H$ is an increasing variance swap kernel. Then the highest model independent lower bound on the price of a variance swap is given by the expression in (\ref{eq:4.2}).
\end{corollary}

\begin{corollary} \label{c:oneprice}
If $\Phi(u,y)$ does not depend on $y$ then the corresponding variance
swap is perfectly replicable by $(\psi,-\psi')$. 
For all consistent models the variation swap has price 
$\int_{\R^+} 
\psi(x) \mu(dx)$. 
\end{corollary}

\begin{example} \label{ex:oneprice} Recall the definitions of the kernels
$H^B$ and $H^Q$ and Example \ref{ex:indepkernel}. $\Phi^B(u,y)=2u^{-2}$
and so $\psi'(u)=-2/u$ and $\psi(u)=-2 \log(u)$. Thus
$H^B(x,y)=\psi(y)-\psi(x)-\psi'(x)(y-x)$ and the strategy $(\psi,-\psi')$
replicates the payoff perfectly for any price realisation. The observation
that $H^B$ has one model-independent price was first made by Bondarenko in
\cite{Bondarenko:07}. Similarly, $H^Q(x,y)=\psi(y)-\psi(x)-\psi'(x)(y-x)$,
where $\psi(x)=x^2$. An alternative analysis of these two payoffs is due to
Neuberger \cite{Neuberger:10}. Neuberger introduces the aggregation property.
Translated into the notation of our setting, a kernel enjoys the aggregation
property if $\E[V_H(X,P^{(n)})]=\E[H(X_T-X_0)]$. Both 
Bondarenko~\cite{Bondarenko:07} and Neuberger~\cite{Neuberger:10}
advocate the use of $H^B$ due to the fact that its price is not sensitive to
the price path, but only to the value of $X_T$. \end{example}

\section{Non-zero interest rates}
\label{sec:nzir}

To date we have worked with forward prices. This has the implication that the 
dynamic part of a hedging strategy has zero cost. In this section we outline 
how 
our analysis can be extended to non-zero, but deterministic, interest rates.

Suppose that interest rates are deterministic. Let $D_t=D_t(T)$ be the 
discount factor over $[t,T]$ so that the asset price realisation ($s=(s_t)_{0 
\leq t \leq T}$) and the forward price realisation are related by $s(t)=D_t 
f(t)$. In the case of constant interest rates $D_t(T)=e^{-r(T-t)}$ so that 
$s(t)=e^{-r(T-t)}f(t)$.

Let $P$ be a partition of $[0,T]$. For $k \in
\{0,1,...,N-1\}$ write $s_k=s(t_k)$, $f_k=f(t_k)$ and $D_k=D_{t_k}(T)$. 
Set $D_{k,k+1}=D_{k+1}/D_k$. Note that if interest
rates are non-negative then $D_{k,k+1} \geq 1$.

Let $G$ be the kernel of a variation swap and write
$G_k(x,y)=G(D_k x,D_k y)$. 
Then the payoff of the 
variance swap is given by  
\[ V_{G}(s,P) = \sum_{k=0}^{N-1} G(D_k f_k, D_{k+1} f_{k+1}) 
= \sum_{k=0}^{N-1} G_k(f_k,D_{k,k+1} f_{k+1}). \]

\begin{proposition} \label{p:decomp}
Suppose that there exists a variation swap kernel $H$, functions $\eta$, 
$\epsilon$, $B$ and a constant $A \in \R$ such that for all $D>0$
\begin{equation} \label{eq:boundinterest}
G_k(x,yD) \geq A H(x,y)+\eta(y)-\eta(x)+\epsilon(x,k,D)(y-x)+B(k,D).
\end{equation}
Without loss of generality we may take $\eta(f(0))=0$.

Suppose that there exists a semi-static sub-hedging strategy $(\psi, 
\Delta)$ for the variation swap with kernel $H$.
Then 
\[V_G(s,P) \geq (A\psi + \eta)(f(T))+\sum_k 
[\epsilon(f_k,k,D_{k,k+1})+ \delta_{t_k}((f(t)_{t \leq 
t_k})](f_{k+1}-f_k) + 
\sum_k 
B(k,D_{k,k+1}),\]
and there is a model-independent sub-hedge and price lower bound for 
$V_G$.
\end{proposition}

\begin{proof}
We have
\begin{eqnarray} \label{eq:bound}
V_{G}(s,P) &=& \sum_{k=0}^{N-1} G_k(f_k,D_{k,k+1} f_{k+1}) \nonumber \\
&\geq& \sum_k [A H(f_k,f_{k+1}) + \eta(f_{k+1})-\eta(f_k) + \epsilon(f_k,k,D_{k,k+1})(f_{k+1}-f_k) + B(k,D_{k,k+1})] \nonumber \\
&\geq& A[\psi(f(T)) + \sum_k \delta_{t_k}((f(t)_{t \leq t_k}) 
(f_{k+1}-f_k)]+ 
\eta(f(T))  
\nonumber \\
&& \hspace{10mm} + \sum_k \epsilon(f_k,k,D_{k,k+1})(f_{k+1}-f_k)+\sum_k 
B(k,D_{k,k+1}) \nonumber 
\end{eqnarray}
\end{proof}

\begin{remark} If we are content to assume that interest rates are 
non-negative then we only need (\ref{eq:boundinterest}) to hold for $D 
\geq 1$. \end{remark}

\begin{remark} The price for the floating leg associated with the hedge
is the price of the static vanilla portfolio with payoff
$(A\psi+\eta)(f(T))$ plus the constant $\sum_{k=0}^{N-1}
B(k,D_{k,k+1})$. \end{remark}

\begin{corollary}
Suppose $H$ is an increasing variance kernel, and $\psi$ is of Class 
$\mathcal K$. 
If 
(\ref{eq:boundinterest}) holds then we have a path-wise sub-hedge and a
model independent bound on the price of $V_G$.
\end{corollary}

In the setting of increasing or decreasing variance kernels
the bound in (\ref{eq:bound}) will be tight provided 
$(\psi,-\psi')$ is 
a tight semi-static hedge for $V_H(f,P)$ and there is equality in 
Equation (\ref{eq:boundinterest}).

\begin{example}\label{eg:HR}
Suppose $G(x,y)=H^R(x,y)=\frac{(y-x)^2}{x^2}$. Then $G_k(x,y)=G(x,y)$, so that 
$\epsilon(x,k,D)$ and $B(k,D)$ will not depend on $k$. Moreover,
\begin{eqnarray*}
G(x,yD) &=& \frac{1}{x^2} (Dy-Dx+Dx-x)^2 \\ 
&=& D^2 \left(\frac{y-x}{x}\right)^2 + D\frac{(D-1)}{x}(y-x)+(D-1)^2 
\end{eqnarray*}
Suppose that interest rates are non-negative so that $D_{k,k+1} \geq 1$. 
Then (\ref{eq:boundinterest}) holds for $A=1$, $\eta=0$, 
$\epsilon(x,D)=D(D-1)/x$ and $B(D)=(D-1)^2$.

Note that there is an inequality in (\ref{eq:boundinterest}) for 
$A=1$. If $D_{k,k+1}$ is independent of $k$ 
(the natural example is to assume
that
interest rates are constant and the partition is uniform, in which case
$d= \log D_{k,k+1} = {rT/N}$)
then we can have equality by 
taking $A=e^{2rT/N}$. In that case we have an improved bound, but the 
improvement becomes negligible in the limit $N \uparrow \infty$.
\end{example}

\begin{example}\label{eg:HL}
Suppose $G(x,y)=H^L(x,y)=(\log(y)-\log(x))^2$. Then $G_k(x,y)=G(x,y)$ 
and $G(x,yD)=(\log 
D+\log y -\log x )^2=H^L(x,y)+2\log D (\log y -\log x)+(\log D)^2.$

Suppose now that the partition is such that $D_{k,k+1}$ is independent 
of $k$, and set $d=\log D_{k,k+1}$. 
Then Equation (\ref{eq:boundinterest}) holds with equality 
for
$A=1$, $\eta(y)=2d\log y$, $\epsilon=0$ and $B(D)=d^2$. 
\end{example}

\begin{example}\label{eg:HB}
Suppose $G(x,y)=H^B(x,y)=-2(\log y -\log x) - (y/x-1))$. Then 
$G_k(x,y)=G(x,y)$ and
\begin{eqnarray*}
G(x,yD) &=& -2(\log y -\log x +\log D)+2D(y-x)+2(D-1) \\
&=& H^B(x,y)+2(D-1)(y/x-1)+H^B(1,D). 
\end{eqnarray*}
Then Equation (\ref{eq:boundinterest}) holds with equality for $A=1$, 
$\eta(y)=0$, 
$\epsilon(x,D)=2(D-1)/x$, $B(D)=H^B(1,D)$.
\end{example}

We can consider the limit as the partition becomes dense, in which case 
the bounds for the variance swap become tight.
For definiteness we will assume that we cave a sequence of uniform 
partitions with mesh size tending to zero, and that interest rates are 
constant, though this can be 
weakened for the squared return and Bondarenko kernels.

Then, for each of the three examples above we have that 
$\sum_{k=0}^{N-1} B(k, D_{k,k+1}) = N B(e^{rT/N}) \rightarrow 0$. 
Further, in each case $\eta(y) \rightarrow 0$, and $A=1$. 
Then 
in the limit the lower bound on the price of the variance 
swap based on the price realisation $s$ is the same as the upper and 
lower bounds for the variance swap defined relative to the forward 
price $f$. Thus, for variance swaps based on frequent monitoring, the bounds 
we have calculated in earlier sections based on the forward price may 
also be used for undiscounted price processes.

\subsection{Super-hedges and upper bounds}
\begin{corollary}
Suppose there exists $H$, $\eta$, $\epsilon$, $B$, and $A$ such that
\begin{equation} \label{eq:boundinterestrev}
G_k(x,yD) \leq A H(x,y)+\eta(y)-\eta(x)+\epsilon(x,k,D)(y-x)+B(k,D),
\end{equation}
and suppose
that there exists a semi-static super-hedging strategy $(\psi,
\Delta)$ for the variation swap with kernel $H$. Then there is a
corresponding  model-independent super-hedge and price upper bound for
$V_G$.
\end{corollary}

The analysis of the kernels $H^R, H^L, H^B$ and upper bounds  
is similar to that in Examples~\ref{eg:HR}---\ref{eg:HB} above. 
For the kernel $H^B$, the choices listed in Example~\ref{eg:HB} give equality 
in (\ref{eq:boundinterestrev})
and can be used equally for upper bounds.
Provided 
that we have an upper bound for $D_{k,k+1}$, so that $D_{k,k+1} \leq \bar{D}$ 
uniformly in $k$, for the kernel $H^R$ 
we may take $A=\bar{D}^2$, $\eta=0$, $\epsilon(x,D)=D(D-1)/x$ and 
$B(D)=(D-1)^2$. Finally, for $H^L$, provided interest rates are non-negative, 
we can write
\[ G(x,yD)
=H^L(x,y)+2\log D (\log y -\log x)+(\log D)^2
\leq H^L(x,y) +2 \frac{\log D}{x} (y-x) +(\log D)^2  \]
so that (\ref{eq:boundinterestrev}) holds for
$A=1$, $\eta=0$, $\epsilon(x,D)= 2 (\log D )/x$ and
$B(D)=(\log D)^2$. Note that, unlike for the lower bound in 
Example~\ref{eg:HL}, for the upper bound we do not need to assume that 
$D_{k,k+1}$ is 
independent of $k$.

\begin{remark} In his analysis of lower bounds for the kernel $H^L$, 
Kahal\'e~\cite{Kahale:11} does not need to assume the partition is uniform and 
that interest rates are constant (or more generally that $D_{k,k+1}$ is 
constant), and can allow for arbitrary finite partitions and deterministic 
interest rates.  Our results complement his results nicely. Although we need 
the assumption that $D_{k,k+1}$ is constant to recover Kahal\'e's result in 
the setting of lower bounds and the kernel $H^L$, in all other cases of 
study (upper 
bounds for $V_{H^L}$ and upper and lower bounds for $V_{H^R}$ and $V_{H^B}$) 
our methods also allow for arbitrary partitions and non-constant but 
deterministic interest 
rates.
\end{remark}

\section{Numerical Results}
\label{sec:numerical}

Given a continuum of call prices, it is possible to calculate the model
independent bounds for the prices of variance swaps. When the implied terminal
distribution of the asset price is simple it is sometimes possible to
calculate the monotone functions associated with the Perkins embedding
explicitly (see Example 5.4) and to obtain a closed form integral expression
for the model independent upper and lower bounds. For more realistic and
complex target laws, the monotone functions and bounds can still be calculated
numerically. The case when the terminal law is lognormally distributed is of
particular practical interest. 

A standard time frame for a volatility swap is
30 days or one month ($T=1/12$), which is the time frame used for the widely 
quoted 
VIX index. Figure~\ref{fig:num} plots the upper and lower bounds 
for the
prices of variance swaps based on the kernels $H_R$ and $H_L$ relative to the
cost of $-2 \log$ contracts (the Neuberger/Dupire price of the standard hedge
or `VIX price') against the volatility parameter of the lognormal (terminal)
distribution centered at $1$. More precisely, the bounds are plots of 
\[ \sigma
\rightarrow {\E[\psi_{\kappa,H}(X_{\sigma/ 
\sqrt{12}})]}/{\E[-2\log X_{\sigma/\sqrt{12}}}], \hspace{10mm} \mbox{and} 
\hspace{10mm}
\sigma
\rightarrow {\E[\psi_{\ell,H}(X_{\sigma/
\sqrt{12}})]}/{\E[-2\log X_{\sigma/ \sqrt{12}}}], \]
where
$X_\sigma \equiv e^{\sigma N - \sigma^2/2}$ is the lognormal random variable 
with volatility parameter
$\sigma$ and $H = H^R$ or $H^L$. Here, $\psi_{K,H}$ is the function given in 
Definition~\ref{d:psi} and $\kappa$ is
chosen according to Proposition~\ref{p:bestkappa} (with $\ell$ chosen 
similarly).
Thus the upper
bound for the kernel $H_L$ and the lower bound for the kernel $H_R$ correspond
to the decreasing function $\kappa$ associated with the Perkins embedding,
while the other two bounds are constructed with the increasing function
$\ell$ associated with the reversed Perkins embedding. 

Note that the price of a variance swap in the Black-Scholes model (as given by 
$\E [ - 2 \log X_{\sigma \sqrt{T}} ])$ is an increasing function of 
volatility. The upper and lower bounds are also increasing functions of 
volatility, and, as can be seen in the figure, they also become wider as 
volatility increases, when expressed as a ratio against the no-jump case. For 
reasonable values of volatility, and for both kernels, the impact of jumps is 
to affect the price by a factor of less than two, and for the kernel $H^L$ the 
bounds are even tighter. The observation that the bounds for the kernel $H_R$ 
are wider than those for the kernel $H_L$ is partly explained by considering 
the leading term in the expansion of the hedging error (see Section 3.2). We 
have $J_R(x) \approx {2x^3}/{3}$ whereas $J_L(x) \approx -{x^3}/{3}$ so that 
the magnitude of the leading error term for $H_R$ is twice that of the leading 
error term for $H_L$. Note that for the optimal martingales the jumps are not 
local, so this approximation becomes less relevant as $\sigma$ increases.

\begin{figure}[H]\label{Fig.3}
\begin{center}
\includegraphics[height=10cm,width=15cm]{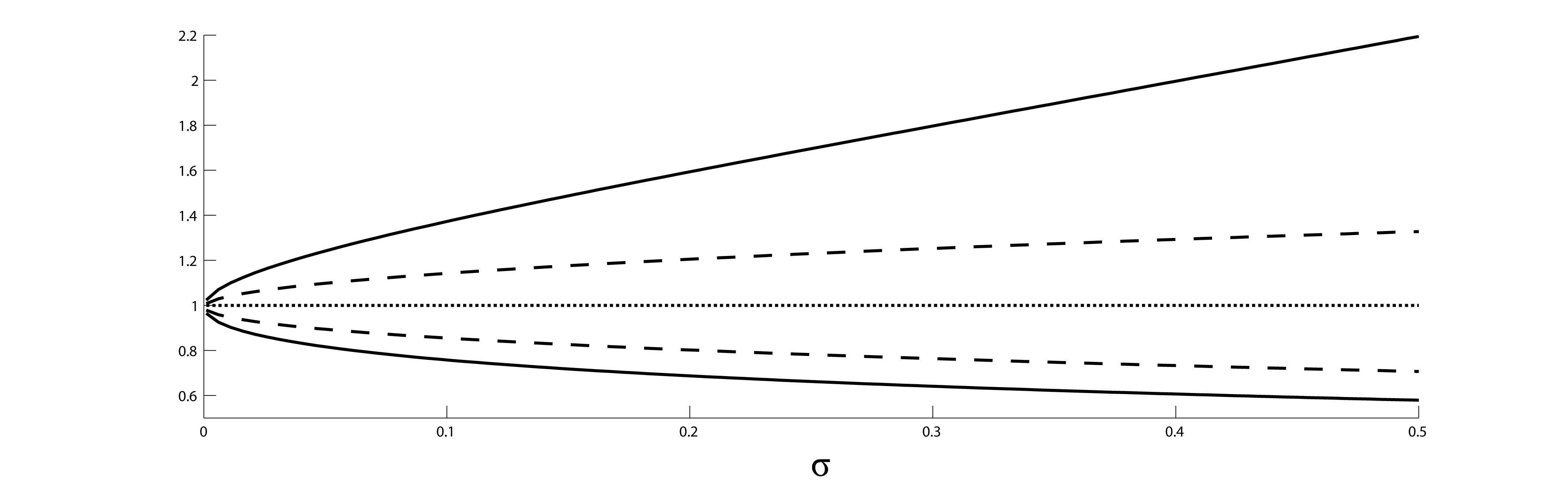}
\caption{Model independent upper and lower bounds for the prices of variance 
swaps based on the kernels $H_L$ (solid lines) and on $H_R$ (dashed lines)
relative to the price of $-2 \log$ contracts (dotted line) in the case when
the terminal distribution is lognormal with volatility between $0$ and $0.5$. 
Here $T=1/12$ and we work with variance swaps on forward prices. }
\label{fig:num}
\end{center}
\end{figure}

\section{Summary and concluding remarks}
\label{sec:remarks}

This article developed from an attempt to express the results of 
Kahal\'{e}~\cite{Kahale:11} on no-arbitrage lower bounds for the prices of 
variance swaps in the framework of model-independent hedging, in which 
extremal 
models and 
prices are associated with extremal solutions of the Skorokhod embedding 
problem. Beginning with Hobson~\cite{Hobson:98}, the focus in this literature 
is on hedging, and on finding pathwise inequalities relating the payoff of the 
exotic, path-dependent derivative and the payoff of a static vanilla call 
portfolio combined with the gains from trade from an investment in the 
underlying security. In the context of variance swaps we find that the lower 
bound is associated 
with a martingale price process which can be expressed as a time-change of the 
Perkins solution of the Skorokhod embedding problem.
This embedding has appeared previously in finance 
in the construction of model-independent bounds for the prices of barrier 
options (Brown et al~\cite{BrownHobsonRogers:01}).

We approach the problem of finding hedging strategies in a more general 
setting than Kahal\'{e}~\cite{Kahale:11} in that we consider a variety 
of kernels in the definition of the
variance swap. The ability to consider general kernels 
allows us to emphasise the dependence of the payoff on the 
presence and character of the jumps, and to show that the nature of this
dependence is strongly influenced by the form of the kernel.
Bondarenko~\cite{Bondarenko:07} and Neuberger~\cite{Neuberger:10} argue that
the finance industry should consider defining variance swaps using the 
kernel $H^B$ as then they
can be replicated perfectly, even in the presence of jumps, recall
Example \ref{ex:oneprice}. The counterargument is that variance swaps provide
value precisely because they are not redundant in this way. Sophisticated
investors want to be able to take positions on the likely presence and
direction of jumps. This is possible if the variance swap is defined using the
kernel $H^R$ or $H^L$, but not using $H^B$.

Kahal\'e~\cite{Kahale:11} only considers the kernel $H^L$, and lower bounds 
and sub-replicating strategies. On the other hand he works directly with the 
undiscounted asset price, and does not give special attention to  
contracts 
written on the 
forward price. He introduces the class of $V$-convex 
functions which have the property that each such function gives a lower bound 
on the price of the variance swap, and an associated sub-hedge. He then 
proceeds to show that functions $\psi$ of Class $\mathcal L$ (in our notation) 
are $V$-convex. In this way he can deduce a lower bound on the price of a 
variance swap. Further, for a particular choice of decreasing function he can 
show that this lower bound can be attained in the continuous time limit under 
a well-chosen stochastic model --- hence the bound he attains must be a best 
bound.

In contrast, initially we consider contracts based on the forward price. This 
simplifies the analysis significantly and reduces the search for candidate 
sub-hedge payoffs to a search for functions satisfying (\ref{eq:psimin}). The 
condition (\ref{eq:psimin}) is considerably simpler than the corresponding 
condition for V-convexity in Kahal\'e~\cite[Equation (3.1)]{Kahale:11}. The 
fact that we have a more transparent representation of the key property allows 
us to find candidate super-hedge payoffs quite easily and allows us to extend 
the analysis to general variation swap kernels provided they have a 
monotonicity property. Moreover, we can easily develop upper bounds to 
complement the lower bounds. Only later do we introduce interest rates and 
variance swaps written on the undiscounted asset price, at which point we find 
simple inequalities which extend our bounds 
to the general case. In the limit of a dense sequence of 
partitions the same bounds are optimal in both the undiscounted and forward 
price settings. We believe that the two-stage approach brings insight, not 
least because in the forward case there is a direct link to martingales and 
solutions of the Skorokhod embedding problem, and because inequalities such as 
(\ref{eq:boundinterest}) allow us to quantify the price difference between 
contacts written on the undiscounted and forward prices for discrete 
monitoring.

A further contribution of this article is to provide a derivation of bounds on 
the prices of variance swaps without any recourse to probability. This 
involves construction of a class of hedges parameterised by monotone 
functions, and the choice of an optimal element in this class for a given set 
of call prices, together with F\"ollmer's non-probabilistic It\^o calculus. 
Price trajectories for which the bound is path-wise tight have at most one 
jump, after which the trajectory is constant. Probability is only required to 
show that these trajectories correspond to a stochastic model for the price 
process. The relationship between the optimality of the cheapest hedge, 
derived in a purely non-proabilistic fashion, and the optimality of the 
Perkins embedding provides a pleasing completeness to the story.

\bibliography{general}
\end{document}